\documentclass[10pt,Journal]{IEEEtran}

\usepackage{amsmath,amssymb,amsfonts,amsthm,mathtools,bm}
\usepackage{algorithm}
\usepackage{algpseudocode}
\usepackage{graphicx}
\usepackage{booktabs}
\usepackage{multirow}
\usepackage{enumitem}
\usepackage{xspace}
\usepackage{hyperref}
\usepackage{subcaption}
\usepackage{array}
\usepackage{color}
\usepackage{cite}
\usepackage{tikz}
\usepackage{stmaryrd}
\usepackage{physics}
\graphicspath{{figures/}} 
\usepackage{caption}
\usepackage{bm}          

\DeclareMathOperator{\Ext}{Ext}
\newtheorem{definition}{Definition}[section]
\newtheorem{theorem}[definition]{Theorem}
\newtheorem{lemma}[definition]{Lemma}
\newtheorem{proposition}[definition]{Proposition}
\newtheorem{corollary}[definition]{Corollary}
\newtheorem{remark}[definition]{Remark}
\newtheorem{example}[definition]{Example}

\newcommand{\Qsys}{\mathcal{Q}}
\newcommand{\PC}{P_C}
\newcommand{\PQ}{P_Q}
\newcommand{\HQ}{\mathcal{H}_Q}
\newcommand{\HA}{\mathcal{H}_A}
\newcommand{\viewA}{\mathrm{view}_A}
\newcommand{\Acc}{\mathrm{Acc}}
\newcommand{\Sec}{\mathsf{Sec}}

\newcommand{\set}[1]{\left\{ #1 \right\}}
\newcommand{\llb}{\llbracket}
\newcommand{\rrb}{\rrbracket}
\newcommand{\Obs}{\mathsf{Obs}}
\newcommand{\SF}{\mathsf{SF}}

\tikzset{
    z spider/.style={circle, draw=black, thick, fill=green!40, minimum size=6mm, inner sep=0pt},
    x spider/.style={circle, draw=black, thick, fill=red!40, minimum size=6mm, inner sep=0pt}
}

\newcommand{\zxdiscard}[1]{
    \draw[thick] (#1) -- ++(0, -0.25) coordinate (#1_g);
    \draw[thick] ([xshift=-0.25cm]#1_g) -- ([xshift=0.25cm]#1_g);
    \draw[thick] ([xshift=-0.15cm, yshift=-0.12cm]#1_g) -- ([xshift=0.15cm, yshift=-0.12cm]#1_g);
    \draw[thick] ([xshift=-0.05cm, yshift=-0.24cm]#1_g) -- ([xshift=0.05cm, yshift=-0.24cm]#1_g);
}

\begin{document}

\title{Current-State Opacity in Safe Partially Observed Quantum Petri Nets: 
True-Concurrency Semantics and Exact Symbolic Verification}

\author{Sichen Ding and Zhiwu Li,~\IEEEmembership{Fellow,~IEEE}
\thanks{Sichen Ding and Zhiwu Li are with the Institute of Systems Engineering, 
Macau University of Science and Technology, Taipa 999078, Macau SAR, 
China (e-mail: 3230006283@student.must.edu.mo; zwli@must.edu.mo).}}

\maketitle

\begin{abstract}
Classical opacity theory for discrete-event systems relies strictly on observable event sequences, 
fundamentally failing to capture security breaches in hybrid architectures where 
an attacker exploits both classical traces and localized quantum correlations. 
To address this gap, we formalize current-state opacity within the framework of 
safe partially observed quantum Petri nets by introducing a true-concurrency semantics that 
represents classical observations as partially ordered multisets via unfolding configurations.
Building upon this, we define quantitative posterior-state leakage as the trace distance between 
the attacker's localized quantum states, evaluated conditionally on whether 
the underlying system has reached a secret or non-secret marking. 
This formulation strictly preserves classical opacity definitions. 
To achieve computational tractability, we apply the stabilizer formalism and develop 
an exact symbolic verification algorithm. By combining targeted unfolding exploration, 
state aggregation exclusively at maximal unobservable reach, and stabilizer-tableau propagation, 
this procedure circumvents both concurrent interleaving explosions and 
exponential density-matrix overhead. Finally, an entanglement-swapping case study validates 
the exact leakage evaluation, demonstrates substantial computational gains, 
and establishes a rigorous interface for counterexample-guided leakage enforcement.
\end{abstract}


\section{Introduction}
\label{sec:introduction}

Opacity is a fundamental confidentiality property for partially observed discrete-event systems (DESs), 
formally requiring that an external observer cannot infer with certainty whether 
secret behavior has occurred or if the current state belongs to a secret set. 
Since the lanuch of foundational state-based formulations of opacity for DESs, 
an extensive body of research has touched upon the modeling, verification, comparison, 
and enforcement of this property across both automata and Petri-net frameworks 
\cite{saboori2007notions,lin2011opacity,saboori2014prob,jacob2016overview,lafortune2018history,wu2013comparative,balun2021comparing,wintenberg2022framework,yin2019stochastic}. 
Concurrently, synthesis-oriented studies have developed a variety of active mitigation strategies, 
including supervisory control, insertion functions, edit functions, 
and dynamic-observation mechanisms for strict opacity enforcement 
\cite{dubreil2010supervisory,wu2014insertion,wu2016optimal,ji2018enforcement,yin2020dynamic,ji2019edit,xie2024optimal}. 
Within this landscape, Petri nets are particularly attractive when concurrency, synchronization, 
and resource sharing dictate system dynamics. 
Consequently, opacity verification for labeled Petri nets has progressively advanced from 
static state-based reachability to sophisticated online and partially observed settings 
\cite{tong2017verification,cong2018online,saadaoui2020current,han2023strong}.

A more demanding setting arises in hybrid quantum-classical networked systems, including 
quantum repeaters, purification-based long-distance communication architectures, and 
quantum-internet platforms 
\cite{briegel1998repeater,dur1999repeater,kimble2008internet,munro2015repeaters,wehner2018internet,cacciapuoti2020internet}. 
In these architectures, asynchronous classical communication and event-driven control logic, such as 
heralding signals and feed-forward corrections, strictly orchestrate the evolution of fragile, 
distributed quantum resources. However, classical opacity theory is fundamentally insufficient for 
this paradigm because it strictly limits the attacker's observation capability 
to classical event logs. 
The hybrid nature of these systems profoundly expands the threat surface: 
An eavesdropper may physically 
possess a subset of the persistent quantum registers, formally designated as the 
\emph{attacker interface}, while 
simultaneously monitoring the public classical control traffic.

Given that routine quantum operations inherently 
distribute correlations, the execution of unobservable internal control branches can establish hidden 
quantum entanglement with the attacker's local interface. This creates 
a critical verification blind spot: 
A classical DES verifier might declare a system perfectly opaque because the secret and non-secret 
executions emit the identical public event sequence. Yet, physically, 
these two structurally equivalent 
executions may leave the attacker's local quantum registers 
in mathematically distinguishable density operators. 
This paradigm shift is profound. Classical observation equivalence at 
the control layer provides 
a false sense of security; to guarantee true confidentiality, 
the criterion for current-state opacity must advance from the set-theoretic ambiguity of 
classical state estimates to the quantitative indistinguishability of localized density operators. 
By rigorously evaluating these posterior density operators conditioned on the 
publicly observable classical event traces, this formulation provides 
a critical theoretical advantage: It successfully unifies the causal rigor of 
discrete-event control theory with the quantitative precision of quantum information theory, 
yielding an exact, physically meaningful measure of information leakage \cite{ying2013quantum}.

Concurrency exacerbates this fundamental challenge. A purely interleaving semantics expands 
a concurrent execution into all of its possible linearizations, thereby duplicating behavior 
that differs only by arbitrary permutations of independent actions. This state-space explosion 
is well recognized in classical concurrency theory, motivating the development of 
partial-order models, event structures, and unfolding-based verification to capture causality, 
conflict, and independence without redundant interleavings 
\cite{pratt1986partialorders,winskel1987event,mcmillan1992using,esparza2002improvement}. 
In the context of opacity, the consequences of interleaving are both conceptual and algorithmic: 
If public observations are represented by raw linear interleavings, the executions 
that are semantically equivalent up to structural independence may be treated as 
distinct observations, even though their divergence carries no security-relevant significance. 
To resolve this, an appropriate mathematical abstraction should combine a true-concurrency semantics 
with an observation model that inherently factors out independent reorderings.
This requirement is particularly critical for quantum Petri nets, 
where a rigorously grounded concurrent semantics has only recently begun to emerge 
\cite{joachim2025qpn}.

To transition from abstract theoretical semantics to a constructive verification framework, 
it is crucial to identify a fragment of the underlying quantum dynamics 
where exact symbolic reasoning remains computationally tractable.
The stabilizer fragment provides an ideal mathematical substrate for this purpose: 
Clifford unitary evolution and projective Pauli measurements admit highly efficient simulation 
and canonical tableau representations, circumventing the exponential overhead of 
global density matrices by confining dense computations strictly to the localized attacker interface 
\cite{aaronson2004improved}. Complementing this algorithmic tractability, 
the ZX-calculus offers a rigorously complete diagrammatic proof language for 
stabilizer quantum mechanics~\cite{backens2014zx,coecke2017picturing}. 
Through its mixed-state extensions, 
it provides a sound topological calculus for completely positive maps and process equivalence 
~\cite{carette2019completeness,wetering2020zx,selinger2007dagger}.
By synthesizing these algorithmic and diagrammatic foundations, true-concurrency opacity analysis for 
quantum Petri nets can transcend purely abstract formalisms to become both semantically rigorous 
and explicitly executable.

To bridge these theoretical and algorithmic strands, we introduce 
\emph{safe partially observed quantum Petri nets} (SPO-QPNs), 
an integrated modeling framework that combines a safe control net with a persistent quantum-register 
layer and a strictly fixed attacker interface. The safe control layer governs causality, 
conflict, synchronization, and resource exclusion, while the persistent quantum-register layer 
ensures that all posterior attacker states are mathematically well-posed over a common, 
invariant Hilbert space physically corresponding to the attacker's accessible interface. 
This architectural design is essential: Without anchoring the local states to 
a fixed physical subset of registers, trace-distance-based leakage quantification is 
mathematically ill-posed across different concurrent executions.


Building upon this structural foundation, we define an event-structure semantics 
in which public observations are formulated not as raw linear interleavings, 
but as rigorously defined observation pomsets (partially ordered multisets), 
which naturally capture concurrent events without imposing artificial ordering.
This yields a true-concurrency semantics for partial observation that perfectly reflects 
the exact commutation behavior of independent quantum branch updates. Within this semantics, 
we formulate qualitative structural current-state opacity and 
quantitative quantum leakage over \emph{normalized} posterior attacker states 
conditioned on an observation pomset and the secret condition. 
The resulting trace-distance leakage metric has a direct operational meaning: Under equal priors, 
it uniquely quantifies the optimal distinguishing advantage between the secret and non-secret 
normalized posteriors, thereby rigorously isolating the genuinely quantum footprint 
from classical probabilistic biases.

To make this theory computationally constructive, we restrict the underlying dynamics to 
the safe stabilizer fragment and develop an exact symbolic verification architecture. 
This true-concurrency engine propagates canonical stabilizer-tableau representations rather than 
exponentially large density matrices, aggregates concurrent executions strictly by 
observation pomsets rather than raw linear interleavings, and computes posterior attacker states 
directly in symbolic form. Complementing this algorithmic layer, we introduce a ZX-calculus 
certificate framework. While it does not drive the primary state-space exploration, 
it distills the symbolic posteriors into compact, rewrite-stable diagrammatic normal forms, 
providing independently verifiable proof objects for posterior-state equivalence and leakage witnesses.

A further contribution is moving beyond exact verification toward automated opacity enforcement. 
Once an observation pomset causing an opacity violation is identified, 
the framework leverages causal hitting sets to synthesize local policy updates. 
This provides a mathematically sound combination of classical flow restriction ($\delta$-updates) 
to prune specific configuration subsets, and invisible quantum state obfuscation ($\mu$-updates) to 
inject local depolarizing noise. The foundational analytical principle driving this synthesis is 
the contractivity of trace distance under completely positive trace-preserving maps, 
which mathematically guarantees that the localized invisible masking provably bounds distinguishability 
without altering the classical visibility semantics. In this light, rather than pursuing a globally 
optimal supervisory-control synthesis, we establish a rigorous, counterexample-guided 
enforcement loop that naturally aligns with the semantic and symbolic machinery developed above.

By systematically integrating these structural modeling, symbolic verification, and 
policy enforcement paradigms, this paper establishes a comprehensive mathematical framework 
for analyzing and securing concurrent quantum discrete-event systems. 
Specifically, the main theoretical and algorithmic contributions are summarized as follows:
\begin{enumerate}[leftmargin=*]
    \item Introduction of safe partially observed quantum Petri nets (SPO-QPNs) 
    as a verification-oriented model for concurrent quantum systems, 
    integrating a safe control net, a persistent quantum-register layer, 
    and a fixed attacker interface.

    \item Establishment of operational and true-concurrency semantics for SPO-QPNs based on 
    unfolding configurations and observable pomsets, coupled with a formal proof that 
    the induced configuration denotations are well-defined and invariant under permutations of 
    independent events.

    \item Formulation of structural opacity and posterior-state leakage under partial observation 
    via attacker posteriors conditioned on observation pomsets, 
    accompanied by a classical conservativity result demonstrating that the proposed framework 
    properly extends classical Petri-net opacity.

    \item Design of an exact symbolic verification procedure for the safe stabilizer fragment, 
    leveraging targeted unfolding exploration, aggregation at maximal unobservable reach, 
    and stabilizer-tableau propagation, supported by a rigorous analysis confirming 
    its correctness, decidability, and computational complexity.

    \item Construction of ZX-calculus certificates for stabilizer-fragment posteriors 
    to serve as an auxiliary certification layer, alongside formal proofs 
    demonstrating that sound diagrammatic rewrites preserve the posterior equivalences 
    and leakage witnesses used in opacity verification.
    
    \item Development of a systematic approach to verification-guided leakage reduction via 
    supervisory restriction and unobservable masking, featuring formal proofs of correctness 
    for the resulting closed-loop enforcement.
    
    \item Application of the theoretical framework to a representative case study, 
    demonstrating that the proposed semantics accurately captures leakage mechanisms 
    induced by true concurrency and quantum side information. Numerical analysis 
    of this case further confirms that the symbolic procedure provides significant 
    computational advantages over interleaving-based density-matrix exploration, 
    and that the exact verification results can successfully guide cost-aware leakage mitigation.
\end{enumerate}

The remainder of this paper is organized as follows. 
Section~\ref{sec:preliminaries} reviews the required mathematical background. 
The SPO-QPN framework is formally introduced in Section~\ref{sec:model}, 
followed by the development of its operational and 
event-structure semantics in Section~\ref{sec:semantics}. 
Based on these theoretical foundations, Section~\ref{sec:opacity} defines quantum opacity under 
partial observation. 
Section~\ref{sec:algorithm} details the exact symbolic verification algorithm tailored to 
the stabilizer fragment. 
As an auxiliary certification layer, the ZX-calculus framework is presented 
in Section~\ref{sec:zx_certificates}, 
while Section~\ref{sec:enforcement} addresses the problem of automated opacity enforcement. 
The theoretical and algorithmic contributions are then validated through 
an entanglement-swapping case study in Section~\ref{sec:case_repeater}. 
Finally, Section~\ref{sec:conclusion} concludes the paper.

%

\section{Preliminaries}
\label{sec:preliminaries}

This section fixes the Petri-net, observation, and quantum-information notations used 
throughout the paper. The presentation is organized in the same order as the later development: 
the underlying safe control-net semantics, 
the observation object compatible with true concurrency, 
the quantum-process semantics, and finally the metric and certificate layer used 
in the opacity analysis.

\subsection{Safe Petri Nets and Their Unfoldings}

The control layer of the model is formalized as a safe place/transition net. 
To this end, we first review the standard safe-net notation, 
followed by the unfolding semantics which characterizes system executions modulo interleavings.

A safe place/transition net is defined as a quadruple $N = (P, T, F, M_0)$, 
where $P$ and $T$ are finite disjoint sets of places and transitions, respectively; 
$F \subseteq (P \times T) \cup (T \times P)$ is the flow relation; 
$M_0 \subseteq P$ is the initial marking. 
Since we restrict our attention to the safe setting, 
every reachable marking $M$ is uniquely characterized by the set of places it marks, 
and thus is identified with a subset $M \subseteq P$. 
For any $x \in P \cup T$, the preset and postset of $x$ are defined by 
${}^{\bullet}x = \{y \in P \cup T \mid (y, x) \in F\}$ 
and $x^{\bullet} = \{y \in P \cup T \mid (x, y) \in F\}$, respectively. 
Adopting the standard contact-free firing rule, 
a transition $t \in T$ is enabled at a marking $M \subseteq P$, 
denoted by $M[t\rangle$, if
$${}^{\bullet}t \subseteq M \quad \text{and} \quad (M \setminus {}^{\bullet}t) \cap t^{\bullet} = \emptyset.$$
The firing of such a transition evolves the marking $M$ into $M'$, 
denoted by $M[t\rangle M'$, where
$$M' = (M \setminus {}^{\bullet}t) \cup t^{\bullet}.$$

Given that the subsequent opacity semantics must distinguish causality from mere interleaving order, 
the concurrent executions of $N$ are modeled via unfoldings rather than sequences. 
Consequently, we utilize occurrence nets to structurally encode the causal history 
and branching choices of these executions.
Formally, an occurrence net is an acyclic net $\mathcal{N}=(B,E,G)$, 
where $B$ is the set of conditions (representing specific occurrences of places in $P$), 
$E$ is the set of events (representing specific occurrences of transitions in $T$), 
and $G \subseteq (B \times E) \cup (E \times B)$ is the flow relation, 
such that each condition $b \in B$ has at most one pre-event (i.e., $|{}^{\bullet}b| \le 1$), 
and the partial order induced by the transitive closure of $G$ is well-founded.
The induced causality relation on events is denoted by \(<\). 
Conflict is the least hereditary symmetric relation generated by immediate conflict, 
where two distinct events are in immediate conflict whenever they share a precondition.

A configuration of \(\mathcal N\) is a finite set \(C\subseteq E\) that is conflict-free 
and downward closed with respect to \(<\). 
Configurations represent finite partially ordered executions. The unfolding of 
a safe net \(N\), denoted by \(\mathcal U(N)\), 
is the canonical branching process associated with \(N\), 
unique up to isomorphism; its configurations represent the executions of \(N\) 
modulo interleavings \cite{mcmillan1992using,esparza2002improvement}.

\subsection{Partially Ordered Multisets and Observation Classes}

The formalization of opacity relies on partial observation. 
However, conventional linear sequences (or observation strings) 
are too restrictive for the concurrent setting investigated in this work. 
Specifically, two executions that differ solely by 
the interleaving of independent observable events should be regarded as observationally equivalent. 
Consequently, partially ordered multisets are adopted as the fundamental observation objects 
to capture this equivalence.

To capture partial observation, the system's events are partitioned via a global labeling mechanism. 
Let $\Sigma_{\mathrm{o}}$ be a finite set of observable labels, 
and $\tau\notin \Sigma_{\mathrm{o}}$ be the unobservable (silent) label. 
Formally, this mechanism is given by a labeling map 
$\lambda: E \to \Sigma_{\mathrm{o}} \cup \{\tau\}$. 
An event $e \in E$ is thus deemed observable if $\lambda(e) \in \Sigma_{\mathrm{o}}$, 
and unobservable if $\lambda(e) = \tau$.
A pomset over \(\Sigma_{\mathrm{o}}\) is defined as an isomorphism class of 
finite labeled partial orders. 
Abstractly, given a finite set \(V\), a partial order \(\preceq\) on \(V\), 
and a labeling map \(\ell:V\to\Sigma_{\mathrm{o}}\), the corresponding pomset is denoted by
\[
[(V,\preceq,\ell)]_{\cong}.
\]
Such two triples \((V,\preceq,\ell)\) and \((V',\preceq',\ell')\) are isomorphic 
if there exists a bijection \(f:V\to V'\) that preserves both labels and order, namely,
\[
\ell'(f(v)) = \ell(v) \quad \text{for all } v \in V,
\]
and
\[
v_1 \preceq v_2 \iff f(v_1) \preceq' f(v_2) \quad \text{for all } v_1, v_2 \in V.
\]

To apply this abstraction to the system's behavior, 
we formalize an individual execution as an event structure. 
Specifically, any configuration \(C\) of the unfolding induces 
a finite labeled event structure \(\mathcal E_C=(C,\preceq,\lambda|_C)\), 
where the partial order \(\preceq\) is the reflexive closure of the strict causality relation \(<\) 
restricted to \(C\), 
and \(\lambda|_C: C \to \Sigma_{\mathrm{o}}\cup\{\tau\}\) is 
the labeling mapping applied to the events in \(C\). 
The set of observable events within this specific execution is defined as
\[
C_{\mathrm{o}}=\{e\in C:\lambda(e)\neq \tau\}.
\]
The observation class induced by the configuration \(C\) is the pomset
\[
\Obs(C)
\coloneqq 
\big[(C_{\mathrm{o}},\,\preceq|_{C_{\mathrm{o}}\times C_{\mathrm{o}}},\,\lambda|_{C_{\mathrm{o}}})\big]_{\cong}.
\]
Two configurations are considered observation-equivalent if they induce the same pomset under \(\Obs\). 
By definition, \(\Obs(C)\) abstracts away unobservable events and identifies the executions 
that differ only by a reordering of concurrent observable events, 
while preserving both observable labels and the observable causal order. 
Crucially, this formalism ensures that the observation semantics, 
and the subsequent opacity analysis, are evaluated locally on individual concurrent executions 
(i.e., configurations) rather than globally on the monolithic unfolding graph. 
In the sequel, the observation class of any configuration \(C\) 
is denoted uniformly by \(\Obs(C)\).

\subsection{Quantum States, Channels, and Instruments}

We next fix the quantum-process notation used for branch semantics. 
Throughout the paper, all Hilbert spaces are finite-dimensional complex Hilbert spaces.

For Hilbert spaces \(\mathcal H\) and \(\mathcal K\), let \(\mathcal L(\mathcal H,\mathcal K)\) 
denote the space of linear operators from \(\mathcal H\) to \(\mathcal K\), 
and write \(\mathcal L(\mathcal H)= \mathcal L(\mathcal H,\mathcal H)\). 
The set of density operators on \(\mathcal H\) is
\[
\mathcal D(\mathcal H)
=
\{\rho\in\mathcal L(\mathcal H): \rho\ge 0,\ \Tr(\rho)=1\},
\]
stands for \(\rho\ge 0\) denotes positive semidefiniteness, 
and \(\Tr(\cdot)\) denotes the trace operation~\cite{nielsen2010quantum}.

A linear map \(\Phi:\mathcal L(\mathcal H)\to\mathcal L(\mathcal K)\) 
is called completely positive (CP) if, for every positive integer \(n\), the ampliation
\[
\Phi\otimes \mathrm{id}_{\mathbb C^n}:
\mathcal L(\mathcal H\otimes\mathbb C^n)\to
\mathcal L(\mathcal K\otimes\mathbb C^n)
\]
maps positive semidefinite operators to positive semidefinite operators, 
where \(\mathbb C^n\) denotes the \(n\)-dimensional complex Hilbert space, 
and \(\mathrm{id}_{\mathbb C^n}\) is the identity map on \(\mathcal L(\mathbb C^n)\).
Furthermore, the map $\Phi$ is said to be trace preserving (TP) if $\Tr(\Phi(X))=\Tr(X)$ 
for every $X\in\mathcal L(\mathcal H)$, 
and trace non-increasing (TNI) if $\Tr(\Phi(X))\le \Tr(X)$ 
for every positive semidefinite operator $X\in\mathcal L(\mathcal H)$.

A quantum channel is a CP and TP map. 
A finite quantum instrument with outcome set \(R\) is a family
\[
\mathcal I=\{\Phi^{(r)}\}_{r\in R}
\]
of CP and TNI maps \(\Phi^{(r)}:\mathcal L(\mathcal H)\to\mathcal L(\mathcal K)\) 
whose sum \(\sum_{r\in R}\Phi^{(r)}\) is a quantum channel~\cite{heinosaari2011mathematical,wilde2013quantum}. 
For an input state \(\rho\in\mathcal D(\mathcal H)\), the probability of outcome \(r\) is
\[
p(r\mid \rho)=\Tr\!\big(\Phi^{(r)}(\rho)\big).
\]
Whenever this probability is strictly positive, the corresponding posterior state 
$\rho_r \in \mathcal D(\mathcal K)$ is given by
$$
\rho_r=
\frac{\Phi^{(r)}(\rho)}{\Tr\!\big(\Phi^{(r)}(\rho)\big)}.
$$

These notions formalize the branch of labelled executions introduced later: 
each branch of a transition instrument induces a CP and TNI map, 
while the sum over all branch outcomes forms a quantum channel.

\subsection{Trace Distance}

The quantitative measure of opacity employed in the subsequent analysis relies on the trace distance 
between the attacker's posterior states. To formalize this metric, 
we first recall the definition of the trace norm. 
For any linear operator $X\in\mathcal L(\mathcal H)$, its trace norm is defined as
\[
\|X\|_1=\Tr\!\big(\sqrt{X^\dagger X}\big).
\]
For \(\rho,\sigma\in\mathcal D(\mathcal H)\), the trace distance between \(\rho\) and \(\sigma\) is
\[
D(\rho,\sigma)=\frac{1}{2}\|\rho-\sigma\|_1.
\]

The trace distance takes values in the interval \([0,1]\) and is contractive under quantum channels: 
if \(\Phi\) is a quantum channel, then
\(D\big(\Phi(\rho),\Phi(\sigma)\big)\le D(\rho,\sigma)\)~\cite{nielsen2010quantum}.
Moreover, when \(\rho\) and \(\sigma\) occur with equal prior probability, 
the optimal success probability for distinguishing them is
\[
P_{\mathrm{succ}}^{\ast}(\rho,\sigma)=\frac{1+D(\rho,\sigma)}{2}.
\]
This operational interpretation motivates the use of trace distance 
as the quantitative leakage measure in Section~\ref{sec:opacity}~\cite{helstrom1976quantum}. 

\subsection{Stabilizer Fragment and ZX Certificates}

While the general opacity framework naturally accommodates arbitrary quantum processes, 
the algorithmic verification developed in this work is deliberately restricted to 
the stabilizer fragment. This restriction guarantees that both symbolic state propagation 
and diagrammatic certification remain exactly computable and scalable, 
avoiding the exponential overhead associated with universal quantum simulation.

Let $\mathcal R$ be a finite set of indexed quantum subsystems. 
The stabilizer fragment on $\mathcal R$, denoted by $\SF(\mathcal R)$, 
is formally defined as the smallest class of quantum processes generated by 
stabilizer-state preparations, Clifford unitary channels, Pauli-basis measurements, 
and the partial trace operation, completely 
closed under sequential and parallel composition \cite{aaronson2004improved}. 
Semantically, every element of $\SF(\mathcal R)$ resolves to either a deterministic quantum channel 
(if all classical measurement outcomes are traced out) or a finite quantum instrument 
(if classical outcomes are explicitly retained for branching).

To rigorously prove observational equivalence without exhaustive state-space enumeration, 
we employ the ZX-calculus as a formal certificate layer. 
Standard translations map stabilizer-fragment processes into mixed-state ZX-diagrams, 
where diagrammatic rewrite rules are strictly sound with respect to their 
underlying CP semantics 
\cite{backens2014zx,wetering2020zx,carette2019completeness,selinger2007dagger}. 
Consequently, if two stabilizer processes can be transformed into each other 
via a sequence of verified ZX-rewrites, this sequence serves as a mathematical certificate 
that they represent identical quantum channels or instruments.

\section{Safe Partially Observed Quantum Petri Nets}
\label{sec:model}

This section establishes the formal definition of the proposed system model.

\begin{definition}[SPO-QPN]
\label{def:spos_qpn}
A safe partially observed quantum Petri net is an 11-tuple
\begin{equation}
\Qsys=(\PC,\PQ,T,F_C,M_0,\rho_0,\{\mathcal{I}_t\}_{t\in T},\Acc,\Lambda,A,\Sec),
\end{equation}
where
\begin{enumerate}[leftmargin=*]
    \item $\PC$ is a finite set of control places;
    \item $\PQ$ is a finite set of persistent quantum-register places, where each $q\in \PQ$ 
    is associated with a fixed finite-dimensional complex Hilbert space $\mathcal{H}_q$, 
    forming the global quantum space
    \begin{equation}
    \HQ = \bigotimes_{q\in \PQ} \mathcal{H}_q;
    \end{equation}
    \item $T$ is a finite set of transitions;
    \item $F_C\subseteq (\PC\times T)\cup(T\times \PC)$ is the control-layer flow relation;
    \item $M_0\subseteq \PC$ is the initial control marking;
    \item $\rho_0 \in \mathcal{D}(\HQ)$ is the global initial quantum state;
    \item $\{\mathcal{I}_t\}_{t\in T}$ is a family of finite quantum instruments. 
    To ensure algorithmic tractability within this verification framework, we hereafter 
    restrict these instruments to the \emph{stabilizer fragment}. Each 
    $\mathcal{I}_t=\{\Phi_t^{(r)}\}_{r\in R_t}$ acts non-trivially only on 
    $\bigotimes_{q\in \Acc(t)} \mathcal{H}_q$ and is implicitly lifted to 
    $\mathcal{L}(\HQ)$ via the tensor product with the identity map on unaccessed quantum registers;
    \item $\Acc:T\to 2^{\PQ}$ assigns to each transition the subset of persistent quantum 
    registers it accesses;
    \item $\Lambda: \{(t,r) \mid t \in T, r \in R_t\} \to \Sigma_{\mathrm{o}}\cup\{\tau\}$ is 
    the joint branch-labeling function, assigning either an observable label 
    from $\Sigma_{\mathrm{o}}$ or the unobservable (silent) label $\tau$ 
    to each specific branch;
    \item $A\subseteq \PQ$ is the fixed attacker interface, defining 
    the attacker's localized Hilbert space
    \begin{equation}
    \HA = \bigotimes_{q\in A} \mathcal{H}_q;
    \end{equation}
    \item $\Sec \subseteq 2^{\PC}$ is the set of markings designated as secret.
\end{enumerate}
\end{definition}

This formulation deliberately separates the classical control and quantum operational layers. 
Control tokens are consumed and produced strictly within the safe net, 
while the quantum state evolves monotonically on a fixed register architecture. 
This separation avoids any mathematical ambiguity between token flow and subsystem creation, 
guaranteeing that all attacker posterior states reside in the strictly invariant Hilbert space $\HA$.

\subsection{Operational Firing Rule}

The global execution state of the system is strictly represented by a pair $(M, \rho)$, 
where $M \subseteq \PC$ is the current control marking and $\rho \in \mathcal{D}(\HQ)$ 
is the global quantum state. A transition $t \in T$ is classically enabled at $M$ if
\begin{equation}
{}^{\bullet}t \subseteq M \qquad \text{and} \qquad (M \setminus {}^{\bullet}t) \cap t^{\bullet} = \emptyset.
\end{equation}
Suppose that $t$ is enabled. The execution of $t$ along a specific measurement branch $r \in R_t$ 
is realizable if and only if its corresponding outcome probability is strictly positive, i.e,
\begin{equation}
\Pr[(t,r)\mid M,\rho] = \Tr\!\big(\Phi_t^{(r)}(\rho)\big) > 0.
\end{equation}
Conditioned on this positive outcome, the system deterministically evolves 
to a new state $(M',\rho')$, where the updated control marking is
\begin{equation}
M'=(M\setminus{}^{\bullet}t)\cup t^{\bullet},
\end{equation}
and the normalized posterior quantum state is
\begin{equation}
\rho' = \frac{\Phi_t^{(r)}(\rho)}{\Tr\!\big(\Phi_t^{(r)}(\rho)\big)}.
\end{equation}
At any point during the execution, the local state accessible to the attacker 
is obtained by taking the partial trace over all quantum registers outside the attacker interface:
\begin{equation}
\viewA(\rho)=\Tr_{\PQ\setminus A}(\rho).
\end{equation}

\begin{definition}[Structural Independence]
\label{def:independence}
Two distinct transitions $t_1,t_2\in T$ are said to be structurally independent if they satisfy 
the following two conditions:
\begin{enumerate}[leftmargin=*]
    \item $({}^{\bullet}t_1 \cup t_1^{\bullet}) \cap ({}^{\bullet}t_2 \cup t_2^{\bullet}) = \emptyset$;
    \item $\Acc(t_1)\cap \Acc(t_2)=\emptyset$.
\end{enumerate}
\end{definition}

Condition 1 ensures that the transitions are structurally conflict-free and 
causally independent within the classical control layer, while Condition 2 restricts that 
they operate on strictly disjoint quantum registers. This joint formulation strictly anchors 
the true-concurrency semantics: it guarantees that any two independent branch operations 
map disjoint tensor factors and inherently commute up to the canonical symmetry of 
tensor products, thereby allowing their corresponding events to be completely unordered 
in a configuration.

\section{Operational and Event-Structure Semantics}
\label{sec:semantics}

\subsection{Linear Executions and Quantum Denotations}

This subsection establishes the semantics of a fully sequentialized system run. 
A branch-labelled execution of $\Qsys$ is a finite alternating sequence of states 
and branch firings, denoted by
\begin{equation}
\pi = (M_0,\rho_0) \xrightarrow{(t_1,r_1)} (M_1,\rho_1) \xrightarrow{(t_2,r_2)} 
\cdots \xrightarrow{(t_n,r_n)} (M_n,\rho_n),
\end{equation}
where each step adheres to the operational firing rule defined in Section~\ref{sec:model}. 

The quantum denotation of the execution $\pi$, 
representing the unnormalized cumulative state evolution, is the composed CP map
\begin{equation}
\llb \pi \rrb = \Phi_{t_n}^{(r_n)} \circ \cdots \circ \Phi_{t_1}^{(r_1)}.
\end{equation}
By convention, the denotation of an empty execution (where $n=0$) is 
the identity superoperator $\mathrm{id}_{\mathcal{L}(\HQ)}$. 
The probability of realizing the execution $\pi$ from the initial state $\rho_0$ is
\begin{equation}
\Pr(\pi) = \Tr\!\big(\llb \pi \rrb(\rho_0)\big).
\end{equation}
Whenever $\Pr(\pi) > 0$, the exact normalized posterior quantum state reached 
at the end of the execution is
\begin{equation}
\rho_\pi = \frac{\llb \pi \rrb(\rho_0)}{\Pr(\pi)}.
\end{equation}

\subsection{Configuration Semantics and Commutation}

Linear executions impose an artificial total order on concurrent events, 
which obfuscates observational equivalence. 
To capture the true-concurrency semantics of $\Qsys$, 
we evaluate the executions over the branch-expanded control net. 
Specifically, let us define the expanded transition set as 
$T_{\mathrm{branch}} = \set{(t, r) \mid t \in T, r \in R_t}$, 
where the pre-set and post-set of each $(t, r)$ are strictly given by 
${}^{\bullet}(t,r) \coloneqq {}^{\bullet}t$ and $(t,r)^{\bullet} \coloneqq t^{\bullet}$, respectively.

Let $\mathcal{U}(\Qsys) = (B, E, G)$ denote the occurrence net representing 
the unfolding of this branch-expanded structure. 
By definition, each event $e \in E$ represents a specific occurrence of a branch, 
denoted by $(t_e, r_e) \in T_{\mathrm{branch}}$. 
Since distinct branches of a given transition possess identical structural pre-sets, 
any distinct events in $E$ representing mutually exclusive measurement outcomes of 
the same transition occurrence necessarily share pre-conditions in $B$, 
which constitutes an immediate conflict by definition.
Consequently, the event is unambiguously assigned the observable label 
$\lambda(e) \coloneqq \Lambda(t_e, r_e)$, 
and its individual quantum denotation is defined exactly as 
the CP map of this underlying branch: $\llb e \rrb \coloneqq \Phi_{t_e}^{(r_e)}$.

As formalized in Section~\ref{sec:preliminaries}, a configuration $C$ of $\mathcal{U}(\Qsys)$ 
is a finite, causally closed, and conflict-free set of event occurrences. 
The observation class of the configuration is exactly the pomset $\Obs(C)$ derived from 
its observable events. 
To define the quantum denotation of the configuration itself, 
let $e_1, e_2, \dots, e_k$ be any linearization (topological sort) of 
the events in $C$ that is compatible with the causal partial order $\preceq$. 
We define:
\begin{equation}
\llb C \rrb = \llb e_k \rrb \circ \cdots \circ \llb e_1 \rrb.
\end{equation}

The algebraic soundness of this definition is predicated on the fact that any two concurrent events 
in a configuration are structurally independent.

\begin{lemma}[Commutation of Independent Branches]
\label{lem:commutation}
If two enabled branches $(t_1,r_1)$ and $(t_2,r_2)$ are structurally independent 
(as per Definition~\ref{def:independence}), their corresponding quantum maps strictly commute:
\begin{equation}
\Phi_{t_2}^{(r_2)} \circ \Phi_{t_1}^{(r_1)} = \Phi_{t_1}^{(r_1)} \circ \Phi_{t_2}^{(r_2)}.
\end{equation}
\end{lemma}

\begin{proof}[Proof of Lemma~\ref{lem:commutation}]
By Definition~\ref{def:independence}, the structural independence of transitions $t_1$ and $t_2$ 
dictates that they operate on strictly disjoint sets of quantum registers, 
i.e., $\Acc(t_1) \cap \Acc(t_2) = \emptyset$. 

Let $Q_1 = \Acc(t_1)$ and $Q_2 = \Acc(t_2)$, and 
let $Q_{\mathrm{rest}} = \PQ \setminus (Q_1 \cup Q_2)$ 
denote the set of all remaining unaccessed registers. 
Up to a canonical permutation of tensor factors, the global Hilbert space 
can be orthogonally decomposed as
\begin{equation*}
\HQ \cong \mathcal{H}_{Q_1} \otimes \mathcal{H}_{Q_2} \otimes \mathcal{H}_{Q_{\mathrm{rest}}},
\end{equation*}
where $\mathcal{H}_{Q_i} = \bigotimes_{q \in Q_i} \mathcal{H}_q$.

According to the operational semantics defined in Section~\ref{sec:model}, 
the local CP map associated with branch $(t_1, r_1)$, 
denoted by $\tilde{\Phi}_{t_1}^{(r_1)}$, 
acts non-trivially only on $\mathcal{L}(\mathcal{H}_{Q_1})$. 
Its implicit lifting to the global operator space $\mathcal{L}(\HQ)$ is exactly the tensor product
\begin{equation*}
\Phi_{t_1}^{(r_1)} = \tilde{\Phi}_{t_1}^{(r_1)} \otimes \mathrm{id}_{Q_2} \otimes 
\mathrm{id}_{Q_{\mathrm{rest}}},
\end{equation*}
where $\mathrm{id}_{Q}$ denotes the identity superoperator on 
$\mathcal{L}(\mathcal{H}_Q)$. Similarly, the lifted quantum map for branch $(t_2, r_2)$ is
\begin{equation*}
\Phi_{t_2}^{(r_2)} = \mathrm{id}_{Q_1} \otimes \tilde{\Phi}_{t_2}^{(r_2)} \otimes 
\mathrm{id}_{Q_{\mathrm{rest}}}.
\end{equation*}

Applying the composition of these two global maps yields:
\begin{align*}
&\Phi_{t_2}^{(r_2)} \circ \Phi_{t_1}^{(r_1)} \\
&= \big( \mathrm{id}_{Q_1} \otimes \tilde{\Phi}_{t_2}^{(r_2)} \otimes \mathrm{id}_{Q_{\mathrm{rest}}} \big) \circ \big( \tilde{\Phi}_{t_1}^{(r_1)} \otimes \mathrm{id}_{Q_2} \otimes \mathrm{id}_{Q_{\mathrm{rest}}} \big) \\
&= (\mathrm{id}_{Q_1} \circ \tilde{\Phi}_{t_1}^{(r_1)}) \otimes (\tilde{\Phi}_{t_2}^{(r_2)} \circ \mathrm{id}_{Q_2}) \otimes (\mathrm{id}_{Q_{\mathrm{rest}}} \circ \mathrm{id}_{Q_{\mathrm{rest}}}) \\
&= \tilde{\Phi}_{t_1}^{(r_1)} \otimes \tilde{\Phi}_{t_2}^{(r_2)} \otimes \mathrm{id}_{Q_{\mathrm{rest}}},
\end{align*}
where the second equality follows directly from the mixed-product property of 
tensor products of linear maps. 
By symmetry, evaluating the composition in the reverse order 
$\Phi_{t_1}^{(r_1)} \circ \Phi_{t_2}^{(r_2)}$ yields the exact same tensor product. 
Consequently, the two global superoperators strictly commute on $\mathcal{L}(\HQ)$.
\end{proof}

\begin{theorem}[Well-Definedness of Configuration Denotation]
\label{thm:config_welldefined}
For every finite configuration $C$ of $\mathcal{U}(\Qsys)$, 
the quantum denotation $\llb C \rrb$ is invariant under the choice of linearization. 
\end{theorem}

\begin{proof}[Proof of Theorem~\ref{thm:config_welldefined}]
Let $L_1$ and $L_2$ be two arbitrary linearizations of the configuration $C$. By definition, 
both $L_1$ and $L_2$ are topological sorts of the finite partially ordered set $(C, \preceq)$. 

According to standard results in the theory of partially ordered sets 
(and Mazurkiewicz trace theory), any two linear extensions of 
a finite poset can be transformed into each other via a finite sequence of 
adjacent transpositions of concurrent elements. Thus, to prove that $\llb C \rrb$ is invariant, 
it suffices to show that its value remains strictly unchanged under the swap of 
any two adjacent concurrent events in a linearization.

Let $L = \dots, e_i, e_{i+1}, \dots$ be a linearization of $C$, 
where the adjacent events $e_i$ and $e_{i+1}$ are completely concurrent 
(i.e., $e_i \not\preceq e_{i+1}$ and $e_{i+1} \not\preceq e_i$). 
Since $e_i$ and $e_{i+1}$ are concurrent in the occurrence net $\mathcal{U}(\Qsys)$, 
they are neither causally related nor in conflict. In the underlying SPO-QPN model, 
this structural concurrency guarantees that their corresponding classical transitions $t_{e_i}$ 
and $t_{e_{i+1}}$ share no pre-places or post-places in the control layer, 
and they access strictly disjoint sets of persistent quantum registers. 
Therefore, $t_{e_i}$ and $t_{e_{i+1}}$ are structurally independent 
in the sense of Definition~\ref{def:independence}.

By Lemma~\ref{lem:commutation}, the strict structural independence of $t_{e_i}$ 
and $t_{e_{i+1}}$ ensures that their quantum branch maps commute:
\begin{align*}
\llbracket e_{i+1} \rrbracket \circ \llbracket e_i \rrbracket &= \Phi_{t_{e_{i+1}}}^{(r_{e_{i+1}})} \circ \Phi_{t_{e_i}}^{(r_{e_i})} \\
&= \Phi_{t_{e_i}}^{(r_{e_i})} \circ \Phi_{t_{e_{i+1}}}^{(r_{e_{i+1}})} = \llbracket e_i \rrbracket \circ \llbracket e_{i+1} \rrbracket
\end{align*}
Consequently, swapping $e_i$ and $e_{i+1}$ in the sequence alters the symbolic string 
but preserves the exact composition of CP maps. Since $L_1$ can be 
reached from $L_2$ through a sequence of such equivalence-preserving adjacent transpositions, 
the cumulative quantum denotation $\llb C \rrb$ is fundamentally independent of 
the chosen linearization.
\end{proof}

\begin{theorem}[True-Concurrency Invariance of Execution Semantics]
\label{thm:true_concurrency}
Let $\pi$ and $\pi'$ be two linear executions of $\Qsys$ that 
represent different linearizations (interleavings) of the same underlying configuration 
$C$ of $\mathcal{U}(\Qsys)$. Then, their quantum and observational semantics are strictly invariant:
\begin{enumerate}[leftmargin=*]
    \item $\llb \pi \rrb = \llb \pi' \rrb$,
    \item $\Pr(\pi) = \Pr(\pi')$,
    \item $\rho_\pi = \rho_{\pi'}$ whenever the probability is strictly positive, and
    \item both executions natively expose the identical observation pomset $\Obs(C)$.
\end{enumerate}
\end{theorem}

\begin{proof}
Since $\pi$ and $\pi'$ are sequential linearizations of the identical partially ordered configuration 
$C$, Theorem~\ref{thm:config_welldefined} directly guarantees that 
their cumulative quantum denotations are equal, 
i.e., $\llb \pi \rrb = \llb C \rrb = \llb \pi' \rrb$. 
Statements (2) and (3) follow immediately by substituting this exact algebraic equality into 
the operational definitions of execution probability and posterior state 
(Section~\ref{sec:semantics}.1). Finally, Statement (4) follows directly from 
the definition of observation classes: since $\text{Obs}(C)$ is defined solely over 
the causal structure of $C$, it is inherently independent of the specific linearization 
$\pi$ or $\pi'$.
\end{proof}

This theorem provides the fundamental semantic justification for defining observations 
as pomset classes over configurations, 
rather than as raw interleaved strings over sequences.

\begin{theorem}[Conservativity over Classical Discrete Event Systems]
\label{thm:classical_conservativity_semantics}
Assume that a given SPO-QPN model $\Qsys$ satisfies the following classicality constraints:
\begin{enumerate}[leftmargin=*]
    \item the initial global state $\rho_0$ is completely diagonal with respect to the 
    fixed computational basis of $\HQ$;
    \item for every branch $(t,r)$, the associated CP map $\Phi_t^{(r)}$ maps 
    diagonal density operators to diagonal density operators, effectively acting as 
    a classical substochastic transition matrix on the diagonal entries;
    \item the attacker interface solely reads the induced classical marginal distribution.
\end{enumerate}
Then, the event-structure and quantum semantics of $\Qsys$ strictly collapse to 
the standard unfoldings and probabilistic semantics of a classical safe partially observed Petri net.
\end{theorem}

\begin{proof}
Under Constraint (1), the initial density operator $\rho_0$ is isomorphic to 
a classical probability distribution vector defined over the computational basis. 
By Constraint (2), the execution of any branch $(t,r)$ applies 
a map $\Phi_t^{(r)}$ that preserves this diagonal structure. 
Consequently, by straightforward induction over the configuration $C$, 
the cumulative denotation $\llb C \rrb$ resolves to a classical substochastic matrix, 
and the posterior state $\rho_C$ remains strictly diagonal for any reachable history. 
In this regime, the trace calculation $\Tr(\cdot)$ algebraically simplifies to 
the $L_1$-norm summation of the probability vector. 
Furthermore, by Constraint (3), the partial trace operation $\Tr_{\PQ \setminus A}(\cdot)$ 
reduces exactly to the statistical marginalization over the unobservable random variables.

Since Lemma~\ref{lem:commutation} and Theorem~\ref{thm:config_welldefined} rely solely on 
the structural disjointness of operations, which identically applies to 
independent classical substochastic updates, the commutation of concurrent events 
is perfectly preserved. Therefore, the tensor-product quantum semantics sheds 
its phase and entanglement dimensions, structurally and behaviorally reducing to 
classical true-concurrency distribution propagation.
\end{proof}

\section{Quantum Opacity Under Partial Observation}
\label{sec:opacity}

Building upon the true-concurrency semantics established in Section~\ref{sec:semantics}, 
this section formalizes the opacity properties of the proposed model over the set of reachable configurations. For any observable pomset $O$, let $\mathcal{C}(O)$ be the set of all finite configurations of the unfolding $\mathcal{U}(\Qsys)$ that yield $O$, namely:
\begin{equation*}
\mathcal{C}(O) = \{C \mid \Obs(C) = O\}.
\end{equation*}
Furthermore, let $M_C \subseteq \PC$ denote the classical control marking reached upon 
the execution of configuration $C$. 
We define two disjoint subsets of $\mathcal{C}(O)$ based on the current-state secrecy of the system:
\begin{equation}
\mathcal{C}_b(O) = \left\{C \in \mathcal{C}(O) \;\middle|\; \mathbf{1}_{\Sec}(M_C) = b\right\}, \qquad b\in\{0,1\},
\end{equation}
where $\mathbf{1}_{\Sec}(\cdot)$ is the standard indicator function for the set of 
secret control markings $\Sec$. Specifically, $b=1$ (resp. $b=0$) indicates that the system reaches 
a secret (resp. non-secret) control marking.

\begin{definition}[Structural Current-State Opacity]
\label{def:structural_opacity}
A system $\Qsys$ is said to be structurally current-state opaque if, for every reachable observation pomset $O$, the following topological condition holds:
\begin{equation}
\mathcal{C}_1(O) \neq \emptyset \implies \mathcal{C}_0(O) \neq \emptyset.
\end{equation}
\end{definition}
Structural opacity merely guarantees the existence of a classical topological alibi within the control layer. However, in the SPO-QPN framework, the attacker concurrently possesses the quantum interface $A$. Even if a non-secret configuration structurally produces the identical observation pomset $O$ ($\mathcal{C}_0(O) \neq \emptyset$), the corresponding reduced quantum states might be perfectly distinguishable, thereby compromising the secret entirely. Furthermore, 
a configuration might have an occurrence probability of zero due to quantum interference. To rigorously capture this hybrid reality, we must elevate the security metric from topological existence to quantum informational distinguishability.

To ensure that the algebraic summation over configurations is mathematically well-defined, we hereafter assume that the underlying safe control net is divergence-free (i.e., it contains no infinite sequences of unobservable transitions). This topological property guarantees that for any finite observation pomset $O$, the corresponding configuration set $\mathcal{C}_b(O)$ is strictly finite.

\begin{definition}[Posterior-State Leakage]
\label{def:posterior_state_leakage}
Given a reachable observation pomset $O$ and an initial global quantum state $\rho_0$, the exact posterior-state leakage $L_{\rho_0}(O)$ is defined as
\begin{equation}
\begin{aligned}
&\ L_{\rho_0}(O) = {}  \\
& \begin{cases} 
\frac{1}{2} \left\| \bar{\sigma}_{O,1}(\rho_0) - \bar{\sigma}_{O,0}(\rho_0) \right\|_1, & \text{if } p(O,0) > 0 \land p(O,1) > 0 \\
1, & \text{if } p(O,0) = 0 \land p(O,1) > 0 \\
0, & \text{if } p(O,1) = 0
\end{cases}
\end{aligned}
\end{equation}
where $p(O,b) = \Tr\bigl(\Omega_{O,b}(\rho_0)\bigr)$ is the joint probability of 
producing observation $O$ and reaching a control marking with secrecy status $b$ 
(where $b=1$ designates a secret marking, and $b=0$ a non-secret one), 
$\bar{\sigma}_{O,b}(\rho_0) = \frac{\Omega_{O,b}(\rho_0)}{p(O,b)}$ 
is the exact normalized posterior attacker state, and 
\begin{equation}
\Omega_{O,b}(\rho_0) = \sum_{C\in\mathcal{C}_b(O)} \viewA\bigl(\llb C \rrb(\rho_0)\bigr)
\end{equation}
is the unnormalized attacker posterior aggregate over the configuration set.
\end{definition}

This trace-distance quantity possesses a rigorous operational interpretation: 
According to the Helstrom measurement bound \cite{helstrom1969quantum}, 
it mathematically quantifies the optimal distinguishing advantage an attacker possesses to discern whether the underlying control mechanism has reached a secret or non-secret marking, explicitly assuming equal prior probabilities for the two normalized scenarios.

\begin{definition}[Quantitative $\epsilon$-Current-State Opacity]
\label{def:epsilon_opacity}
Given a security threshold $\epsilon \in [0,1]$, a system $\Qsys$ is said to be $\epsilon$-current-state opaque from the initial state $\rho_0$ if the worst-case posterior-state leakage across all reachable observation pomsets $O$ satisfies:
\begin{equation}
\sup_{O} L_{\rho_0}(O) \le \epsilon.
\end{equation}
\end{definition}

\begin{definition}[Robust Upper-Bound Leakage]
\label{def:robust_upper_bound}
For any observation pomset $O$ such that $p(O,0)>0$ and $p(O,1)>0$, the robust upper-bound leakage $U_{\rho_0}(O)$ is defined as
\begin{equation}
U_{\rho_0}(O) = \sup_{\sigma_1 \in \mathcal{V}_1(O),\, \sigma_0 \in \mathcal{V}_0(O)} \frac{1}{2}\|\sigma_1 - \sigma_0\|_1,
\end{equation}
where for each $b \in \{0,1\}$, the set
\begin{equation*}
\begin{aligned}
\mathcal{V}_b(O) = {} 
\bigg\{ \viewA & \left( \frac{\llbracket C \rrbracket(\rho_0)}{\Tr(\llbracket C \rrbracket(\rho_0))} \right) \;\bigg|\; 
\begin{aligned}
& C \in \mathcal{C}_b(O) \ \land \\
& \Tr(\llbracket C \rrbracket(\rho_0)) > 0
\end{aligned} \bigg\}
\end{aligned}
\end{equation*}
represents the collection of all reachable, normalized localized states accessible to the attacker, conditioned on the underlying system reaching either a secret ($b=1$) or non-secret ($b=0$) control marking.
\end{definition}

\begin{proposition}
\label{prop:leq_upper_bound}
For every observation pomset $O$ with $p(O,0)>0$ and $p(O,1)>0$, it holds:
\begin{equation}
0 \le L_{\rho_0}(O) \le U_{\rho_0}(O) \le 1.
\end{equation}
\end{proposition}

\begin{proof}[Proof of Proposition~\ref{prop:leq_upper_bound}]
The bounds $0 \le L_{\rho_0}(O)$ and $U_{\rho_0}(O) \le 1$ hold trivially since 
the trace distance between any two density operators strictly falls within 
the interval $[0, 1]$. To prove the inequality $L_{\rho_0}(O) \le U_{\rho_0}(O)$, 
observe that the exact normalized posterior state $\bar{\sigma}_{O,b}(\rho_0)$ is 
a convex combination of the normalized localized states 
$\sigma_C = \viewA\left( \frac{\llb C \rrb(\rho_0)}{\Tr(\llb C \rrb(\rho_0))} \right)$ 
generated by each reachable configuration $C \in \mathcal{C}_b(O)$ 
(i.e., configurations with $\Tr(\llb C \rrb(\rho_0)) > 0$). 
The corresponding convex weights are $w_C = \frac{\Tr(\llb C \rrb(\rho_0))}{p(O,b)}$, 
which inherently satisfy $\sum_{C} w_C = 1$.

We can rewrite the exact leakage as a double sum over both configuration sets:
\begin{equation*}
L_{\rho_0}(O) = \frac{1}{2}\left\| \sum_{C_1} \sum_{C_0} w_{C_1} w_{C_0} (\sigma_{C_1} - \sigma_{C_0}) \right\|_1,
\end{equation*}
where the summations range over the reachable configurations in $\mathcal{C}_1(O)$ 
and $\mathcal{C}_0(O)$, respectively. By applying the triangle inequality (the subadditivity of 
the trace norm) and factoring out the non-negative scalar weights, one obtains:
\begin{align*}
L_{\rho_0}(O) &\le \sum_{C_1} \sum_{C_0} w_{C_1} w_{C_0} \left( \frac{1}{2} \|\sigma_{C_1} - \sigma_{C_0}\|_1 \right) \\[1.5ex]
&\le \sum_{C_1} \sum_{C_0} w_{C_1} w_{C_0} \left( \sup_{\substack{\sigma'_1 \in \mathcal{V}_1(O) \\ \sigma'_0 \in \mathcal{V}_0(O)}} \frac{1}{2} \|\sigma'_1 - \sigma'_0\|_1 \right) \\[1.5ex]
&= U_{\rho_0}(O) \sum_{C_1} w_{C_1} \sum_{C_0} w_{C_0} = U_{\rho_0}(O).
\end{align*}
This completes the proof.
\end{proof}

\begin{corollary}[Opacity Conservativity]
\label{cor:opacity_conservativity}
Under the assumptions of Theorem~\ref{thm:classical_conservativity_semantics}, 
qualitative current-state opacity reduces exactly to classical current-state opacity, 
and the leakage $L_{\rho_0}(O)$ reduces to the total-variation distance between 
the classical posterior marginals induced by the observation pomset $O$.
\end{corollary}

\begin{proof}[Proof of Corollary~\ref{cor:opacity_conservativity}]
By Theorem~\ref{thm:classical_conservativity_semantics}, under the stated classicality constraints, 
the global quantum state remains diagonal in the computational basis throughout any execution, 
effectively serving as a classical probability vector. 
The CP maps collapse into classical substochastic matrices, 
and the partial trace operation $\viewA(\cdot)$ analytically corresponds to 
marginalizing the joint probability distribution over the unobservable quantum registers.

Consequently, the unnormalized posterior aggregates $\Omega_{O,b}(\rho_0)$ algebraically represent 
the unnormalized classical marginal distributions. For any diagonal density matrices $X$ and $Y$, 
the trace distance $\frac{1}{2}\|X - Y\|_1$ is algebraically identical to 
the classical total variation distance between their diagonal distributions. 
Therefore, $L_{\rho_0}(O)$ exactly computes the total variation distance between 
the normalized classical conditional probability distributions, 
thus perfectly recovering the standard probabilistic opacity metric for 
classical discrete event systems.
\end{proof}

\begin{remark}[Channel-Level Extension]
A strictly stronger security notion can be formulated by replacing the posterior states with 
posterior effective quantum channels, and evaluating the distinguishability via 
the induced diamond norm rather than the state-level trace distance. 
We do not pursue this stronger route in the present framework simply because 
the state-level metric already possesses a rigorous physical meaning (via the Helstrom bound), 
acts as an exact measure within the target stabilizer fragment, 
and provides a sufficient foundation for algorithmic verification and enforcement.
\end{remark}

\section{Symbolic Verification in the Stabilizer Fragment}
\label{sec:algorithm}

Building upon the semantic definitions of structural current-state opacity 
and posterior-state leakage established in Section~\ref{sec:opacity}, 
this section addresses the algorithmic challenges of evaluating these metrics. 
While the configuration-based framework sidesteps classical interleavings theoretically, 
conventional state-space exploration would still succumb to the interleaving explosion, 
compounded by the exponential overhead of quantum density-matrix evolution. 
To achieve tractability without sacrificing analytical exactness, 
this section develops a symbolic verification architecture tailored to 
the stabilizer fragment.
To rigorously employ the standard Gottesman-Knill theorem and 
its associated exact data structures within this architecture, we hereafter restrict 
the algorithmic focus to qubit systems (i.e., $\mathcal{H}_q \cong \mathbb{C}^2$ for all $q \in \PQ$).
By coupling true-concurrency exploration directly with this qubit-based 
symbolic state propagation, the proposed framework computes the exact opacity metrics 
while effectively circumventing both dimensions of state-space explosion.

\subsection{Targeted Exploration and Unobservably-Closed Space}
\label{subsec:verification_scope}

To ensure finite termination without sacrificing analytical exactness, 
we restrict our analysis to a finite target family of observation pomsets, 
denoted by $\mathcal O_{\mathrm{tar}}$. 
Recall that for every finite configuration $C$ of 
the unfolding $\mathcal U(\Qsys)$, $\Obs(C)$ denotes its observation pomset. 
Let $\mathrm{Pref}(\mathcal O_{\mathrm{tar}})$ denote the pomset-prefix closure of 
$\mathcal O_{\mathrm{tar}}$, comprising all finite observation pomsets obtained by 
restricting an element of $\mathcal O_{\mathrm{tar}}$ to a downward-closed set of observable events.

We define the targeted operational space as the set of configurations 
whose observations remain consistent with some target prefix:
\[
\widehat{\mathcal C}_{\mathrm{tar}} \coloneqq \{\,C \mid \Obs(C)\in\mathrm{Pref}(\mathcal O_{\mathrm{tar}})\,\}.
\]
This constitutes an algorithmic pruning domain rather than a semantic redefinition. 
We assume that $\mathcal O_{\mathrm{tar}}$ is finite and that the underlying Petri net is divergence-free with respect to unobservable events (i.e., no infinite sequences of silent transitions exist). Under this structural assumption, the configuration space $\widehat{\mathcal C}_{\mathrm{tar}}$ is strictly finite.

Let $\mathcal{C}(\mathcal U(\Qsys))$ denote the set of all finite configurations of the unfolding. For any configuration $C\in\widehat{\mathcal C}_{\mathrm{tar}}$, the valid extension frontier is defined as:
\begin{multline*}
\mathrm{Ext}_{\mathrm{tar}}(C) \coloneqq \\
\left\{ e \notin C \;\middle|\; C \cup \{e\} \in \mathcal{C}(\mathcal{U}(\Qsys)) \land \mathrm{Obs}(C \cup \{e\}) \in \mathrm{Pref}(\mathcal{O}_{\mathrm{tar}}) \right\}
\end{multline*}
\subsection{Maximal Reach and Analytical Exactness}
\label{subsec:maximal_reach}

A fundamental challenge in evaluating current-state opacity under true concurrency is determining 
the precise causal frontier for evaluating the secret while strictly preserving the probability mass. 
Accumulating posterior quantum states indiscriminately across transient unobservable prefixes 
mathematically results in the erroneous double-counting of overlapping causal histories.

To ensure strict analytical correctness, 
the quantum density must be aggregated exclusively at the \emph{maximal unobservable reach} 
of each observation class. For any $O \in \mathcal{O}_{\mathrm{tar}}$, 
we define its maximal reachable configuration set within the explored space as:
\begin{equation}
\begin{aligned}
& \mathcal{C}_{\max}(O) = \\
& \quad \left\{ C \in \widehat{\mathcal C}_{\mathrm{tar}} \;\middle|\; \Obs(C) = O \land \nexists e \in \Ext_{\mathrm{tar}}(C) : \lambda(e) = \tau \right\}.
\end{aligned}
\end{equation}
Reachable configurations in $\mathcal{C}_{\max}(O)$ represent the causally deepest executions that yield $O$ 
without triggering further unobservable events, 
accurately reflecting the attacker's localized knowledge state once the observation has 
formally ``settled''---defined strictly by the predicate:~ 
$\forall e \in \Ext_{\mathrm{tar}}(C): \lambda(e) \neq \tau$.

Crucially, any two distinct reachable configurations in $\mathcal{C}_{\max}(O)$ contain 
mutually conflicting events in the underlying occurrence net $\mathcal{U}(\Qsys)$, 
arising from either classical structural conflicts or orthogonal quantum measurements. 
They therefore represent mutually exclusive execution trajectories. 
As disjoint events within the global sample space, these trajectories yield 
unnormalized density matrices with additive probability masses. 
Consequently, the algebraic summation over $\mathcal{C}_{\max}(O)$ exactly aggregates 
these mutually exclusive sub-states without double-counting probability mass.

\subsection{Exact Symbolic Propagation Engine}
\label{subsec:exact_backend}
To computationally evaluate the configuration semantics without explicitly constructing exponentially 
large density matrices, the verification framework utilizes an exact symbolic representation 
rooted in the Gottesman-Knill theorem \cite{aaronson2004improved}. 
Within this framework, every unnormalized global quantum state generated by the stabilizer fragment 
is strictly encoded as a symbolic object $\Gamma$, which mathematically manifests as 
a \emph{weighted mixed-state stabilizer tableau}. 
Let $\mathrm{Den}(\Gamma) \in \mathcal{L}(\HQ)$ denote the explicit density operator mathematically 
represented by the tableau $\Gamma$. The symbolic propagation engine is formally defined by 
the following four exact algebraic operations defined over the space of stabilizer tableaux:
\begin{enumerate}[label=(\roman*), leftmargin=*]
    \item \textbf{Initialization:} A basis tableau $\Gamma_0$ with 
    $\mathrm{Den}(\Gamma_0) = \rho_0$.

    \item \textbf{Symbolic Transition:} A deterministic update function that applies 
    the branch instrument directly to the tableau structure, 
    ensuring that the symbolic and concrete semantics commute:
    \[
    \mathrm{Den}\bigl(\mathrm{Update}(\Gamma,(t,r))\bigr) = \Phi_t^{(r)}\bigl(\mathrm{Den}(\Gamma)\bigr).
    \]
    By the operational properties of the stabilizer formalism, this symbolic transition requires at most 
    $\mathcal{O}(|\PQ|^2)$ algebraic steps, 
    thereby circumventing the exponential overhead of Hilbert-space matrix multiplication.

    \item \textbf{Trace Evaluation:} An algebraic scalar function 
    $w(\Gamma) = \Tr(\mathrm{Den}(\Gamma))$ 
    that retrieves the embedded probability mass of the tableau.
    This yields the exact execution probability $\Pr(\pi)$ for 
    the corresponding causal history without floating-point approximation.

    \item \textbf{Symbolic Partial Trace:} A localized reduction map $\mathrm{Reduce}_A(\Gamma)$ 
    that algebraically traces out the unobservable registers from the tableau, yielding 
    the attacker's reduced state:
    \[
    \mathrm{Den}_A\bigl(\mathrm{Reduce}_A(\Gamma)\bigr) = \viewA\!\bigl(\mathrm{Den}(\Gamma)\bigr),
    \]
    where $\mathrm{Den}_A(\cdot)$ denotes the mapping to operators exclusively on $\HA$.
\end{enumerate}

Grounded in the Gottesman-Knill theorem~\cite{gottesman1998heisenberg_arxiv}, 
the trace evaluation $w(\Gamma)$ operates via exact rational arithmetic, 
thereby eliminating cumulative floating-point errors during state-space pruning. 
Furthermore, since the caching mechanism identifies configurations by their structural history 
rather than the specific algebraic representation of $\Gamma$, 
the state-space aggregation remains strictly sound. 
This approach circumvents the canonicalization problem inherent in non-unique stabilizer tableau 
representations.

\subsection{Exact Targeted Symbolic Exploration Algorithm}
\label{subsec:exact_exploration}

Building upon the symbolic propagation engine, 
this subsection presents the complete algorithmic procedure for computing 
the exact posterior aggregates required to verify opacity over 
the targeted space. The proposed approach, formalized in Algorithm~\ref{alg:exact_target_exploration},
structurally decouples the exploration of the configuration space from the aggregation of 
the probability mass.

To formalize the outputs, for each target observation pomset $O$ and 
secret condition $b \in \{0,1\}$, the algorithm computes three exact objects: 
a boolean reachability flag $S_{O,b}$, the unnormalized posterior aggregate $\Omega_{O,b}$, 
and a reachable witness configuration $W_{O,b}$ utilized for downstream topological certification.

Specifically, Phase 1 systematically unfolds the targeted causal structures 
by maintaining an active worklist $\mathcal{W}$ of unexplored state triplets, 
denoted as $(C, M, \Gamma)$. To strictly prevent redundant evaluations of 
structurally equivalent interleavings, the algorithm globally caches all discovered 
configurations in a structural visited set $\mathcal{V}$. To explicitly isolate 
reachable trajectories, it concurrently maintains a reachable set $\mathcal{V}_{\mathrm{reach}}$ 
for configurations with strictly positive quantum probability, alongside an associative 
dictionary $\mathsf{Cache}$ that maps each reachable configuration $C$ directly to its 
terminal classical marking and quantum tableau, denoted by $\mathsf{Cache}[C] = (M, \Gamma)$. 
Subsequently, Phase 2 explicitly leverages the reachable topology within $\mathcal{V}_{\mathrm{reach}}$ 
and the cached states in $\mathsf{Cache}$ to isolate the configurations at 
the maximal unobservable reach. It then performs an exact accumulation 
of the unnormalized density matrices, fundamentally avoiding the double-counting of probability mass.

\begin{algorithm}[t]
\caption{Exact Targeted Symbolic Exploration and Aggregation}
\label{alg:exact_target_exploration}
\begin{algorithmic}[1]
\Require Target observation family $\mathcal O_{\mathrm{tar}}$, extension function $\Ext_{\mathrm{tar}}$, and exact symbolic propagation engine
\Ensure Exact flags $\{S_{O,b}\}$, posterior aggregates $\{\Omega_{O,b}\}$, and witnesses $\{W_{O,b}\}$

\State \textbf{Phase 1: Configuration Space Exploration and Symbolic Tracking}
\State Initialize structural visited set $\mathcal V \gets \{\emptyset\}$, reachable set $\mathcal V_{\mathrm{reach}} \gets \{\emptyset\}$
\State Initialize worklist $\mathcal W \gets \{(\emptyset, M_0, \Gamma_0)\}$, tracking dictionary $\mathsf{Cache} \gets \{ \emptyset : (M_0, \Gamma_0) \}$
\While{$\mathcal W \neq \emptyset$}
    \State Extract $(C, M, \Gamma)$ from $\mathcal W$
    \ForAll{$e \in \Ext_{\mathrm{tar}}(C)$}
        \State $C' \gets C \cup \{e\}$
        \If{$C' \notin \mathcal V$}
            \State $\mathcal V \gets \mathcal V \cup \{C'\}$ 
            \State Let $e$ correspond to the branch $(t_e, r_e)$
            \State $\Gamma' \gets \mathrm{Update}(\Gamma, (t_e, r_e))$
            \If{$w(\Gamma') > 0$} 
                \State $M' \gets (M \setminus {}^{\bullet}t_e) \cup t_e^{\bullet}$
                \State $\mathcal V_{\mathrm{reach}} \gets \mathcal V_{\mathrm{reach}} \cup \{C'\}$
                \State $\mathsf{Cache}[C'] \gets (M', \Gamma')$
                \State Insert $(C', M', \Gamma')$ into $\mathcal W$
            \EndIf
        \EndIf
    \EndFor
\EndWhile

\State \textbf{Phase 2: Maximal Unobservable Reach Aggregation}
\ForAll{$O \in \mathcal O_{\mathrm{tar}}$ and $b \in \{0,1\}$}
    \State $S_{O,b} \gets 0$, \quad $W_{O,b} \gets \bot$
    \State $\Omega_{O,b} \gets 0_{\mathcal L(\HA)}$ 
\EndFor

\ForAll{$O \in \mathcal O_{\mathrm{tar}}$}
    \State Compute $\mathcal{C}_{\max}(O) \gets \{ C \in \mathcal V_{\mathrm{reach}} \mid \Obs(C) = O \land {}$
    \Statex \hfill $\nexists e \in \Ext_{\mathrm{tar}}(C) : \lambda(e) = \tau \land C\cup\{e\} \in \mathcal V_{\mathrm{reach}} \}$
    \ForAll{$C \in \mathcal{C}_{\max}(O)$}
        \State $(M_C, \Gamma_C) \gets \mathsf{Cache}[C]$
        \State $b_C \gets \mathbf{1}_{\Sec}(M_C)$
        \State $S_{O,b_C} \gets 1$
        \If{$W_{O,b_C} = \bot$} 
            \State $W_{O,b_C} \gets C$ 
        \EndIf
        \State $\Omega_{O,b_C} \gets \Omega_{O,b_C} + \mathrm{Reduce}_A(\Gamma_C)$ 
    \EndFor
\EndFor
\State \Return $\{S_{O,b}\}, \{\Omega_{O,b}\}, \{W_{O,b}\}$
\end{algorithmic}
\end{algorithm}

Upon termination of Algorithm~\ref{alg:exact_target_exploration}, 
the system is structurally opaque over the targeted space if and only if 
$S_{O,1} = 1$ implies $S_{O,0} = 1$ for all $O \in \mathcal{O}_{\mathrm{tar}}$. 
Furthermore, to evaluate $\epsilon$-current-state opacity, 
the exact posterior-state leakage $L_{\rho_0}(O)$ is computed algebraically from 
the unnormalized aggregates $\Omega_{O,b}$ by normalizing the states and evaluating 
their trace distance. The system is rigorously verified as $\epsilon$-opaque over 
$\mathcal{O}_{\mathrm{tar}}$ if $\sup_{O \in \mathcal{O}_{\mathrm{tar}}} L_{\rho_0}(O) \le \epsilon$.

\subsection{Algorithmic Correctness and Complexity}
\label{subsec:opacity_evaluation}

Having detailed the exact targeted symbolic exploration algorithm, 
we now formalize its theoretical guarantees. In this subsection, 
we first establish the finite termination and semantic exactness of 
the algorithmic outputs, followed by an analysis of its computational complexity.

\begin{theorem}[Termination and Semantic Exactness]
\label{thm:targeted_exploration_sound_complete}
If the underlying SPO-QPN is divergence-free with respect to unobservable events, 
Algorithm~\ref{alg:exact_target_exploration} terminates in finite steps. 
Upon termination, for every target observation $O \in \mathcal{O}_{\mathrm{tar}}$ 
and secret condition $b \in \{0,1\}$, 
the returned aggregate $\Omega_{O,b}$ exactly coincides with the theoretical posterior aggregate 
$\Omega_{O,b}(\rho_0)$ defined in Definition~\ref{def:posterior_state_leakage}.
\end{theorem}

\begin{proof}
The finite termination of the algorithm follows directly from the structural assumptions. 
By hypothesis, the SPO-QPN is divergence-free with respect to unobservable events. 
Coupled with the finiteness of the target observation family $\mathcal{O}_{\mathrm{tar}}$, 
this guarantees that the prefix-closed configuration space $\widehat{\mathcal C}_{\mathrm{tar}}$ 
is strictly finite. Consequently, during Phase 1, the monotonic accumulation of configurations 
in the structural set $\mathcal V$ ensures that the worklist $\mathcal{W}$ is depleted 
in a finite number of steps. Phase 2 then iterates strictly over the finite sets 
$\mathcal{O}_{\mathrm{tar}}$ and $\mathcal V_{\mathrm{reach}}$, thereby guaranteeing global termination.

Regarding semantic exactness, the condition $C' \notin \mathcal V$ in Phase 1 ensures that 
every configuration $C \in \widehat{\mathcal C}_{\mathrm{tar}}$ is explored 
and cached exactly once. This uniquely identifies each execution by its causal partial order, 
structurally avoiding the redundant evaluation of concurrent interleavings. 
By the exactness of the symbolic propagation engine, 
the terminal tableau $\Gamma_C$ stored in $\mathsf{Cache}[C]$ satisfies 
$\mathrm{Den}(\Gamma_C) = \rho(C)$, where $\rho(C)$ is the true unnormalized density matrix 
associated with configuration $C$. 

In Phase 2, the algorithm filters $\mathcal V_{\mathrm{reach}}$ 
using the structural extension conditions 
to isolate the maximal causal frontier $\mathcal{C}_{\max}(O)$. 
For each configuration $C \in \mathcal{C}_{\max}(O)$ satisfying the secret predicate 
$\mathbf{1}_{\Sec}(M_C) = b$, the localized attacker view is algebraically extracted via 
$\mathrm{Reduce}_A(\Gamma_C)$. The final accumulated operator is therefore:
\begin{equation*}
    \Omega_{O,b} = \sum_{\substack{C \in \mathcal{C}_{\max}(O) \\ \mathbf{1}_{\Sec}(M_C) = b}} 
    \mathrm{Den}_A\bigl(\mathrm{Reduce}_A(\Gamma_C)\bigr).
\end{equation*}
Applying the symbolic identity $\mathrm{Den}_A(\mathrm{Reduce}_A(\Gamma_C)) = \viewA(\rho(C))$ 
to the summation yields exactly the theoretical posterior aggregate $\Omega_{O,b}(\rho_0)$ 
defined in Section~\ref{sec:opacity}, establishing the exactness of the algorithm.
\end{proof}

The next subsection briefly explains why the procedure is exact over the targeted
space and summarizes its computational cost.

\subsection{Exactness, Decidability, and Complexity}
\label{subsec:targeted_exactness_complexity}

The procedure formalized in Algorithm~\ref{alg:exact_target_exploration} is exact over 
the targeted configuration space. Under the assumptions that the plant is divergence-free and 
the target observation family $\mathcal O_{\mathrm{tar}}$ is finite, 
the targeted configuration space $\widehat{\mathcal C}_{\mathrm{tar}}$ is finite. 
By caching each generated configuration under a canonical representation and processing it 
at most once, the worklist exploration strictly terminates.
For every explored configuration $C$, the associated symbolic object stored in 
$\mathsf{Cache}[C]$ is obtained by the exact propagation of the branch maps along $C$. 
By the configuration-level well-definedness established previously, 
this symbolic object is independent of the chosen linearization of $C$. 
Furthermore, the pruning condition $w(\Gamma_C)>0$ removes those trajectories 
whose quantum probability vanishes, retaining precisely the reachable ones. 
As a result, the set $\mathcal V_{\mathrm{reach}}$ computed by the algorithm is 
exactly the targeted reachable configuration space.

The subsequent aggregation phase preserves this exactness.
For each target observation $O\in\mathcal O_{\mathrm{tar}}$, 
every maximal reachable configuration $C\in\mathcal C_{\max}(O)$ 
evaluates to a unique secret condition $b_C=\mathbf{1}_{\Sec}(M_C)\in\{0,1\}$. 
Therefore, the boolean flag $S_{O,b}$ records the existence of a maximal reachable configuration 
for a given observation and secret status, while $W_{O,b}$ stores a valid witness whenever 
such a configuration exists. Concurrently, the aggregate $\Omega_{O,b}$ computes 
the exact finite sum of the localized operators $\mathrm{Reduce}_A(\Gamma_C)$ 
across all such mutually exclusive configurations. It follows that the targeted verification of 
structural current-state opacity over $\mathcal O_{\mathrm{tar}}$ 
is decidable via a finite check on the flags $\{S_{O,b}\}$, 
and the quantitative leakage $L_{\rho_0}(O)$ is explicitly computable 
for every $O\in\mathcal O_{\mathrm{tar}}$ from the exact aggregates 
$\Omega_{O,0}$ and $\Omega_{O,1}$.

We next summarize the computational cost. Let
\[
d_{\max}\coloneqq \max_{C\in\widehat{\mathcal C}_{\mathrm{tar}}}|C|,
\qquad
\Delta_{\mathrm{tar}}
\coloneqq 
\max_{C\in\widehat{\mathcal C}_{\mathrm{tar}}}
|\Ext_{\mathrm{tar}}(C)|,
\]
and let
\[
|\mathcal C_{\max}|
\coloneqq 
\sum_{O\in\mathcal O_{\mathrm{tar}}} |\mathcal C_{\max}(O)|.
\]
Assume that configurations are stored by canonical keys of length at most
$d_{\max}$, that each symbolic update
$\mathrm{Update}(\Gamma,(t,r))$ costs $\mathcal O(|\PQ|^2)$, and that observation pomsets
together with the existence of reachable $\tau$-extensions are recorded during
exploration. The structural cache ensures that the symbolic update is evaluated 
at most once per distinct configuration. Consequently, the total cost of the 
targeted exploration phase is bounded by
\[
\mathcal O\!\left(
|\widehat{\mathcal C}_{\mathrm{tar}}| (\Delta_{\mathrm{tar}} d_{\max} + |\PQ|^2)
\right).
\]
This bound also absorbs the identification of the maximal sets
$\mathcal C_{\max}(O)$, since the corresponding observation pomsets and maximality
flags can be maintained during exploration without changing the asymptotic order.

The remaining cost comes from attacker-side postprocessing. For each maximal
configuration $C\in\mathcal C_{\max}(O)$, the contribution
$\mathrm{Reduce}_A(\Gamma_C)$ is accumulated into an explicit operator on
$\mathcal H_A$. Since $\dim(\mathcal H_A)=2^{|A|}$, each dense attacker-side
operator has $4^{|A|}$ entries, and the total aggregation cost is
\[
\mathcal O\!\left(|\mathcal C_{\max}|\,4^{|A|}\right).
\]
Finally, if the leakage evaluation is performed by Hermitian diagonalization of
the explicit $2^{|A|}\times 2^{|A|}$ posterior difference matrix, then each
target observation pomset contributes an additional
$\mathcal O(8^{|A|})$ cost, yielding
\[
\mathcal O\!\left(|\mathcal O_{\mathrm{tar}}|\,8^{|A|}\right)
\]
for the final leakage computation stage. Therefore, under the above
implementation assumptions, the overall running time is bounded by
\[
\mathcal O\!\left(
|\widehat{\mathcal C}_{\mathrm{tar}}| (\Delta_{\mathrm{tar}} d_{\max} + |\PQ|^2)
+
|\mathcal C_{\max}|\,4^{|A|}
+
|\mathcal O_{\mathrm{tar}}|\,8^{|A|}
\right).
\]

Two remarks are worth emphasizing. First, the exponential dependence on $|A|$
arises only in the final attacker-side matrix postprocessing and does not affect
the exactness of the symbolic targeted exploration itself. Second, if one is
strictly interested in structural current-state opacity, the Boolean flag computation 
$\{S_{O,b}\}$ suffices, rendering both the $4^{|A|}$ explicit matrix aggregation 
and the $8^{|A|}$ leakage-evaluation terms completely unnecessary.

\section{ZX Certificates for Posterior Equivalence}
\label{sec:zx_certificates}

The symbolic engine of Section~\ref{sec:algorithm} is the primary verification
mechanism of the paper. The purpose of this present section is narrower:
it provides an auxiliary certificate layer that turns already-computed
posterior operators into independently checkable graphical proof artifacts.
Accordingly, the ZX layer does not participate in state-space exploration,
does not alter the decidability result, and is not used to derive the
complexity bound.

Building upon the strict qubit restriction established for the symbolic engine 
in Section~\ref{sec:algorithm}, every persistent quantum register is structurally treated as 
a pure two-level system ($\mathcal{H}_q \cong \mathbb{C}^2$). 
Consequently, the standard mixed-state stabilizer ZX-calculus 
\cite{coecke2017picturing, carette2019completeness, wetering2020zx} 
is strictly complete for this certification phase. 
To formally link these diagrams to the underlying quantum operators,
let $\llb \cdot \rrb_{\mathrm{ZX}}$ denote the standard denotational semantics 
mapping a ZX-diagram to its corresponding CP map. 
Extensions of this specific certification layer to other local dimensions require formalizing 
the corresponding complete stabilizer graphical languages, which is left for future work.

\subsection{Configuration-Level Compilation}

We first establish the formal compilation of individual causal trajectories into diagrammatic objects. 
To ground this construction, we assume a base dictionary of graphical generators: 
For each measurement branch $(t,r) \in T_{branch}$ in $\Qsys$, 
there exists a predetermined mixed-state ZX-diagram $\mathcal{Z}_{t,r}$ 
such that its denotational semantics satisfies 
$\llb \mathcal{Z}_{t,r} \rrb_{\mathrm{ZX}} = \Phi_t^{(r)}$.

\begin{definition}[Configuration Compilation]
\label{def:config_compilation}
Let $C$ be a finite configuration of the unfolding $\mathcal{U}(\Qsys)$, 
and $\pi = e_1 \dots e_k$ be a valid linearization of $C$. 
The sequential diagrammatic compilation along $\pi$, 
denoted $\mathcal{Z}_{C}^{\pi}$, is defined via spatial sequential composition:
\[
\mathcal{Z}_{C}^{\pi} \coloneqq \mathcal{Z}_{t_{e_k},r_{e_k}} \circ_{\mathrm{ZX}} \cdots 
\circ_{\mathrm{ZX}} \mathcal{Z}_{t_{e_1},r_{e_1}}.
\]
\end{definition}

\begin{proposition}[Well-Definedness and Semantic Soundness]
\label{prop:zx_compile_configuration}
Let $C$ be a finite configuration. For any two linearizations $\pi_1$ and $\pi_2$ of $C$, 
their compiled diagrams are equivalent in the mixed-state stabilizer ZX-calculus:
\[
\mathcal{Z}_{C}^{\pi_1} \stackrel{\mathrm{ZX}}{\equiv} \mathcal{Z}_{C}^{\pi_2}.
\]
Moreover, for any linearization $\pi$ of $C$, the denotational semantics of the compiled 
diagram satisfies:
\[
\llb \mathcal{Z}_{C}^{\pi} \rrb_{\mathrm{ZX}} = \llb C \rrb.
\]
\end{proposition}
\begin{proof}
For any chosen linearization $\pi = e_1 \dots e_k$ of $C$, 
the functorial nature of the diagrammatic composition guarantees that 
$\llb \mathcal{Z}_{C}^{\pi} \rrb_{\mathrm{ZX}} = \Phi_{t_{e_k}}^{(r_{e_k})} 
\circ \cdots \circ \Phi_{t_{e_1}}^{(r_{e_1})} = \llb \pi \rrb$. 
By Theorem~\ref{thm:config_welldefined}, if $\pi$ and $\pi'$ are two distinct linearizations of 
the same configuration $C$, their cumulative CP maps are strictly equivalent, 
i.e., $\llb \pi \rrb = \llb \pi' \rrb$. Consequently, their respective compiled diagrams yield 
identical denotations ($\llb \mathcal{Z}_{C}^{\pi} \rrb_{\mathrm{ZX}} = \llb \mathcal{Z}_{C}^{\pi'} \rrb_{\mathrm{ZX}}$). 
By the completeness of the mixed-state stabilizer ZX-calculus for CP maps 
\cite{carette2019completeness}, these two diagrams are provably equivalent within the calculus 
(i.e., $\mathcal{Z}_{C}^{\pi} \stackrel{\mathrm{ZX}}{\equiv} \mathcal{Z}_{C}^{\pi'}$). 
The semantic soundness $\llb \mathcal{Z}_{C}^{\pi} \rrb_{\mathrm{ZX}} = \llb C \rrb$ 
then follows directly from the definition of configuration denotations.
\end{proof}

By virtue of Proposition~\ref{prop:zx_compile_configuration}, 
the compiled diagrammatic representation is intrinsically invariant under 
the choice of linearization. We may therefore unambiguously omit the superscript 
and denote the configuration-level diagram simply as $\mathcal{Z}_C$ in all subsequent analysis.

\subsection{Posterior Certificates and Rewrite Invariance}

While Proposition~\ref{prop:zx_compile_configuration} provides a diagram for 
a single configuration, the unnormalized attacker posterior $\Omega_{O,b}(\rho_0)$ 
computed by Algorithm~\ref{alg:exact_target_exploration} is an algebraic aggregate over 
a set of maximal configurations $\mathcal{C}_{\max}(O)$. 
Since the global initial state $\rho_0$ and all branch updates fall strictly within 
the stabilizer fragment, the resultant operator $\Omega_{O,b}(\rho_0)$ represents 
an unnormalized stabilizer mixed state on the attacker interface $\HA$.

By the universality of the ZX-calculus for the stabilizer fragment, there rigorously exists 
a ZX-diagram that synthesizes this aggregated algebraic state without 
requiring the formal diagrammatic summation of concurrent histories.

\begin{definition}[Posterior Certificate]
\label{def:posterior_certificate}
Let $O$ be an observation pomset and $b \in \{0,1\}$ be a secret condition. 
A well-formed mixed-state ZX-diagram $\mathcal{Z}_{O,b}$ is said to be a posterior certificate for $(O,b)$ 
if it satisfies the following two conditions:
\begin{enumerate}
    \item The spatial boundaries of $\mathcal{Z}_{O,b}$ correspond exactly to 
    the attacker interface $\HA$;
    \item Its denotational semantics matches the unnormalized posterior 
    state, i.e.,
    \begin{equation}
        \label{eq:posterior_certificate}
        \llb \mathcal{Z}_{O,b} \rrb_{\mathrm{ZX}} = \Omega_{O,b}(\rho_0).
    \end{equation}
\end{enumerate}
Such a certificate is conventionally denoted by $\mathcal{Z}_{O,b}$.
\end{definition}

Let the relation $\vdash_{\mathrm{ZX}}$ denote provable equality derived via 
a finite sequence of sound topological rewrite rules within the mixed-state stabilizer ZX-calculus.

\begin{proposition}[Rewrite Invariance]
\label{prop:zx_rewrite_invariance}
Let $\mathcal{Z}_{O,b}$ be a posterior certificate for $(O,b)$. 
If $\mathcal{Z}_{O,b} \vdash_{\mathrm{ZX}} \widehat{\mathcal{Z}}_{O,b}$, 
then $\widehat{\mathcal{Z}}_{O,b}$ is also a posterior certificate for $(O,b)$, such that 
\[
\llb \widehat{\mathcal{Z}}_{O,b} \rrb_{\mathrm{ZX}} = 
\llb \mathcal{Z}_{O,b} \rrb_{\mathrm{ZX}} = \Omega_{O,b}(\rho_0).
\]
\end{proposition}

\begin{proof}
Let $\mathcal{Z}_{O,b}$ be a posterior certificate. By Definition~\ref{def:posterior_certificate}, 
it satisfies $\llb \mathcal{Z}_{O,b} \rrb_{\mathrm{ZX}} = \Omega_{O,b}(\rho_0)$. 
The ZX-calculus is known to be sound with respect to its denotational semantics in the category of 
Hilbert spaces and CP maps~\cite{backens2014zx, coecke2017picturing}. 
Specifically, for any two diagrams $\mathcal{G}, \widehat{\mathcal{G}}$, 
the relation $\mathcal{G} \vdash_{\mathrm{ZX}} \widehat{\mathcal{G}}$ implies the semantic identity 
$\llb \mathcal{G} \rrb_{\mathrm{ZX}} = \llb \widehat{\mathcal{G}} \rrb_{\mathrm{ZX}}$. 
Consequently, the rewritten diagram $\widehat{\mathcal{Z}}_{O,b}$ preserves the algebraic aggregate:
\begin{equation}
    \llbracket \widehat{\mathcal{Z}}_{O,b} \rrbracket_{\mathrm{ZX}} = \llbracket \mathcal{Z}_{O,b} \rrbracket_{\mathrm{ZX}} = \Omega_{O,b}(\rho_0).
\end{equation}
Since $\widehat{\mathcal{Z}}_{O,b}$ maintains the same spatial boundaries as $\mathcal{Z}_{O,b}$ 
under standard rewrite rules, it satisfies all conditions of Definition~\ref{def:posterior_certificate}, completing.
\end{proof}

\begin{corollary}[Zero-Leakage Certificate]
\label{cor:zx_zero_leakage_certificate}
Let $O$ be an observation pomset with $p(O,b) > 0$ for $b \in \{0,1\}$, 
and let $\widehat{\mathcal{Z}}_{O,0}, \widehat{\mathcal{Z}}_{O,1}$ be their respective 
posterior certificates. If there exists a scalar $\alpha > 0$ such that
\begin{equation}
    \label{eq:zx_equivalence_condition}
    \widehat{\mathcal{Z}}_{O,1} \vdash_{\mathrm{ZX}} \alpha \cdot \widehat{\mathcal{Z}}_{O,0},
\end{equation}
then the normalized posterior states coincide, 
i.e., $\bar{\sigma}_{O,1}(\rho_0) = \bar{\sigma}_{O,0}(\rho_0)$, 
and the posterior-state leakage vanishes:
$L_{\rho_0}(O) = 0$.
\end{corollary}

\begin{proof}
By Proposition~\ref{prop:zx_rewrite_invariance} and the soundness of the ZX-calculus, 
the relation \eqref{eq:zx_equivalence_condition} implies the semantic identity:
\begin{equation}
    \Omega_{O,1}(\rho_0) = \alpha \Omega_{O,0}(\rho_0).
\end{equation}
Recall that the normalized posterior state is defined as 
$\bar{\sigma}_{O,b}(\rho_0) = \Omega_{O,b}(\rho_0) / \mathrm{Tr}(\Omega_{O,b}(\rho_0))$. 
Since $\alpha > 0$ and the operators are non-null by the assumption $p(O,b) > 0$, we have:
\begin{equation}
    \bar{\sigma}_{O,1}(\rho_0) = \frac{\alpha \Omega_{O,0}(\rho_0)}{\mathrm{Tr}(\alpha \Omega_{O,0}(\rho_0))} 
    = \frac{\alpha \Omega_{O,0}(\rho_0)}{\alpha \mathrm{Tr}(\Omega_{O,0}(\rho_0))} = \bar{\sigma}_{O,0}(\rho_0).
\end{equation}
The identity $\bar{\sigma}_{O,1}(\rho_0) = \bar{\sigma}_{O,0}(\rho_0)$ implies that the attacker's 
posterior distribution over the secret remains equal to the prior (or independent of the secret $b$). 
Consequently, $L_{\rho_0}(O) = 0$ follows as a direct consequence of 
Definition~\ref{def:posterior_state_leakage}.
\end{proof}

\subsection{The Operational Role of ZX Certificates}
\label{subsec:zx_summary}

To summarize, the theoretical value of the ZX certificate layer lies in its ability to 
fundamentally decouple \emph{proof generation} from \emph{proof verification}. 
While the exact symbolic engine (Section~\ref{sec:algorithm}) efficiently mitigates 
the concurrent state-space explosion to compute the exact algebraic posteriors, 
its execution traces and internal tableau aggregates are challenging to independently audit.

By formally lifting these algebraic aggregates into the strictly complete mixed-state stabilizer 
ZX-calculus, this section transforms the algebraic equivalence of posterior attacker states into 
a purely topological rewrite problem. Ultimately, this equips the verification framework with 
\emph{independently checkable witnesses}: an external auditor can rigorously certify quantum opacity 
(e.g., zero leakage via Corollary~\ref{cor:zx_zero_leakage_certificate}) by mechanically validating 
a sequence of sound diagrammatic rewrites ($\vdash_{\mathrm{ZX}}$), entirely eliminating the need to 
trust or re-execute the underlying concurrent state-space exploration.

\section{Opacity Enforcement via Supervisory Disabling and Invisible Masking}
\label{sec:enforcement}

While the preceding framework provides exact verification of posterior-state leakage, 
securing a system often requires actively mitigating the identified opacity violations. 
To address this, this section introduces a targeted enforcement mechanism that 
combines the classical supervisory disabling of controllable transitions with 
invisible quantum masking operations. Rather than developing a comprehensive 
supervisory-control theory, the objective is to establish a mathematically rigorous interface 
through which verification results can directly guide leakage reduction.

\subsection{Controlled System and Enforcement Policies}

Let
\[
T = T_c \uplus T_{uc}
\]
be the partition of plant transitions into controllable and uncontrollable
transitions. In addition, let $T_m$ be a finite set of auxiliary masking
transitions, disjoint from $T$, and define the extended transition set
\[
\overline T \coloneqq T \uplus T_m .
\]
Each masking transition $t_m \in T_m$ is assumed to satisfy the following statements:
(i) it is triggerable by the supervisor; (ii) it is invisible to the attacker,
in the sense that every branch of $t_m$ carries observation label $\tau$;
(iii) its quantum branch semantics is a CPTP map from the chosen stabilizer
fragment; (iv) it is compatible with the admissibility specification
$\mathcal A_{\mathrm{spec}}$.

The controlled architecture is specified by
\[
\Qsys^{\mathrm{ctrl}}
=
(\Qsys, T_c, T_{uc}, T_m, \mathcal A_{\mathrm{spec}}, \mathcal C_{\mathrm{cost}}),
\]
where $\mathcal A_{\mathrm{spec}}$ is the admissibility specification and
$\mathcal C_{\mathrm{cost}}$ is the cost model for control actions.

We define policies over the finite configurations of the extended uncontrolled unfolding associated 
with the extended transition set $\overline T$, so that the policy domain is fixed independently of 
the eventual closed-loop behavior.

\begin{definition}[Enforcement Policy]
\label{def:enforcement_policy}
An enforcement policy over the extended uncontrolled unfolding is defined as a pair 
$\pi_{\mathrm{ctrl}} = (\delta, \mu)$, where for each finite configuration $C$:
\begin{itemize}[leftmargin=*]
    \item $\delta(C) \subseteq T_c$ is the set of controllable plant transitions disabled at $C$;
    \item $\mu(C) \subseteq T_m$ is the set of masking transitions permitted for injection at $C$.
\end{itemize}
\end{definition}

By design, uncontrollable transitions $T_{uc}$ can never be disabled. 
Given an enforcement policy $\pi_{\mathrm{ctrl}}$, the \emph{closed-loop system}, 
denoted by $\Qsys \parallel \pi_{\mathrm{ctrl}}$, is constructed as the sub-unfolding induced by 
retaining all uncontrollable events, deleting events generated by 
the disabled controllable transitions in $\delta$, and retaining only those masking events 
explicitly permitted by $\mu$.
Furthermore, a policy $\pi_{\mathrm{ctrl}}$ is said to be \emph{admissible} 
if its induced closed-loop system satisfies the prescribed admissibility specification:
\begin{equation}
\label{eq:admissibility}
    \Qsys \parallel \pi_{\mathrm{ctrl}} \models \mathcal A_{\mathrm{spec}}.
\end{equation}

\subsection{Feasible Enforcement Objective}

For an enforcement policy $\pi_{\mathrm{ctrl}}$, let
$L_{\rho_0}^{\pi_{\mathrm{ctrl}}}(O)$ denote the posterior-state leakage
associated with an observation pomset $O$ in the closed-loop system
$\Qsys \parallel \pi_{\mathrm{ctrl}}$.

\begin{definition}[$\epsilon$-feasible enforcement]
\label{def:epsilon_feasible_enforcement}
An admissible policy $\pi_{\mathrm{ctrl}}$ is said to be $\epsilon$-feasible if
\[
\sup_{O \in \Obs(\Qsys \parallel \pi_{\mathrm{ctrl}})}
L_{\rho_0}^{\pi_{\mathrm{ctrl}}}(O) \le \epsilon ,
\]
where $\Obs(\Qsys \parallel \pi_{\mathrm{ctrl}})$ symbolizes the set of
observation pomsets reachable in the closed-loop system.
\end{definition}

Among all $\epsilon$-feasible policies, one may further minimize the cumulative
cost $J(\pi_{\mathrm{ctrl}})$ induced by $\mathcal C_{\mathrm{cost}}$. In the
present paper, however, our algorithmic goal is sound and feasible synthesis rather
than provably optimal synthesis.

\subsection{Contractivity of Posterior Distinguishability under Masking}

The analytical justification for quantum masking lies in the contractivity of 
trace distance under CPTP maps. 
However, since masking transitions operate on the global plant, ensuring 
trace-distance reduction exclusively at the attacker's localized interface 
requires that the masking operation should not inadvertently transfer hidden 
information from the environment into the observable subsystem.

\begin{lemma}[Contractivity under Localizable Masking]
\label{lem:posterior_contractivity_masking}
Let $O$ be a reachable observation pomset. Suppose that, at every configuration $C \in \mathcal{C}(O)$, 
the enforcement policy $\pi_{\mathrm{ctrl}}$ applies a common invisible masking transition 
$t_m \in T_m$ with global CPTP semantics $\mathcal{M}$. 
Assume that $\mathcal{M}$ is independent of the secret condition 
$b \in \{0,1\}$ and is \emph{localizable} to the attacker subsystem $A$, i.e., 
there exists a local CPTP map $\mathcal{M}_A$ on $\mathcal{L}(\mathcal{H}_A)$ such that
\begin{equation}
    \label{eq:localizability_condition}
    \Tr_{\PQ \setminus A} \circ \mathcal{M} = \mathcal{M}_A \circ \Tr_{\PQ \setminus A}.
\end{equation}
Then, the posterior-state trace distance satisfies:
\begin{equation}
    \label{eq:trace_distance_contractivity}
    D\!\left(
    \bar{\sigma}^{\pi_{\mathrm{ctrl}}}_{O,1}(\rho_0),
    \bar{\sigma}^{\pi_{\mathrm{ctrl}}}_{O,0}(\rho_0)
    \right)
    \le
    D\!\left(
    \bar{\sigma}_{O,1}(\rho_0),
    \bar{\sigma}_{O,0}(\rho_0)
    \right).
\end{equation}
\end{lemma}

\begin{proof}
By hypothesis, the common global masking channel $\mathcal{M}$ is applied at every 
configuration $C \in \mathcal{C}(O)$ and is independent of $b$. The closed-loop posterior 
states are obtained by tracing out the unobservable environment $E = \PQ\setminus A$ 
from the post-masking global branch states. Applying the localizability condition 
\eqref{eq:localizability_condition}, we obtain:
\begin{align*}
    \bar{\sigma}^{\pi_{\mathrm{ctrl}}}_{O,b}(\rho_0) 
    &= \Tr_{E}\big(\mathcal{M}(\bar{\rho}_{O,b}(\rho_0))\big) \\
    &= \mathcal{M}_A\big(\Tr_{E}(\bar{\rho}_{O,b}(\rho_0))\big) \\
    &= \mathcal{M}_A\big(\bar{\sigma}_{O,b}(\rho_0)\big),
\end{align*}
where $\bar{\rho}_{O,b}(\rho_0)$ represents the global state immediately prior to masking. 
The inequality \eqref{eq:trace_distance_contractivity} follows directly from the 
contractivity of the trace distance under the local quantum channel $\mathcal{M}_A$.
\end{proof}

The lemma establishes that the localizable invisible masking is strictly non-expansive with respect to 
the attacker's posterior distinguishability. It therefore provides a sound foundational mechanism 
for bounding and subsequently reducing leakage, entirely without altering the functional semantics of 
the observed pomset.

\begin{proposition}[Complete twirling yields zero leakage]
\label{prop:complete_twirling}
Under the assumptions of Lemma~\ref{lem:posterior_contractivity_masking},
suppose that the induced local map $\mathcal{M}_A$ on the attacker 
subsystem is the complete twirling channel, which maps every localized state 
to the maximally mixed state, i.e.,
\[
\mathcal{M}_A(\omega) = I_A/d_A
\qquad
\text{for all } \omega \in \mathcal{L}(\mathcal{H}_A),
\]
where $d_A$ is the dimension of $\mathcal{H}_A$. Then, conditional on the 
observation pomset $O$, the closed-loop leakage vanishes:
\[
L_{\rho_0}^{\pi_{\mathrm{ctrl}}}(O) = 0.
\]
\end{proposition}

\begin{proof}
Since the attacker's classical observation $O$ is identical for both secret 
conditions by definition of the posterior analysis, the distinguishability relies 
entirely on the quantum states. The local map $\mathcal{M}_A$ maps both pre-masking 
posterior states $\bar{\sigma}_{O,1}(\rho_0)$ and $\bar{\sigma}_{O,0}(\rho_0)$ to 
the identical maximally mixed state $I_A/d_A$. Hence, their trace distance is exactly zero.
\end{proof}

A standard, less aggressive alternative to complete twirling is partial depolarization. 
For an arbitrary quantum subsystem $\mathcal{H}_A$ of dimension $d_A$, the generalized 
depolarizing channel acting on any state $\omega \in \mathcal{D}(\mathcal{H}_A)$ is defined as
\[
\mathcal E_p(\omega) = (1-p)\omega + p\frac{I_A}{d_A},
\qquad p \in [0,1].
\]

\begin{proposition}[Leakage reduction via generalized depolarization]
\label{prop:generalized_depolarization}
Under the assumptions of Lemma~\ref{lem:posterior_contractivity_masking}, suppose 
the enforcement policy applies an invisible masking transition that induces the 
generalized local depolarizing channel $\mathcal E_p$ on the attacker subsystem $\mathcal{H}_A$. 
Then, for any observation pomset $O$, the closed-loop leakage scales proportionally:
\[
L_{\rho_0}^{\pi_{\mathrm{ctrl}}}(O) = (1-p) D\big(\bar{\sigma}_{O,1}(\rho_0), \bar{\sigma}_{O,0}(\rho_0)\big).
\]
\end{proposition}

\begin{proof}
Let $\sigma_b \coloneqq \bar{\sigma}_{O,b}(\rho_0)$ for $b \in \{0,1\}$ denote the 
pre-masking posterior states in $\mathcal{D}(\mathcal{H}_A)$. Applying the generalized 
local channel $\mathcal E_p$ to both states, the maximally mixed noise term $p(I_A/d_A)$ 
cancels out exactly, yielding
\[
\mathcal E_p(\sigma_1) - \mathcal E_p(\sigma_0) = (1-p)(\sigma_1 - \sigma_0).
\]
Consequently, by the linearity of the trace operator, the trace distance scales proportionally:
\[
L_{\rho_0}^{\pi_{\mathrm{ctrl}}}(O)
= D\big(\mathcal E_p(\sigma_1), \mathcal E_p(\sigma_0)\big)
= (1-p)\,D(\sigma_1, \sigma_0),
\]
which establishes the general scaling law for any finite-dimensional attacker subsystem.
\end{proof}

\begin{example}[Analytic Evaluation and Threshold Enforcement]
To illustrate the operational impact of Proposition~\ref{prop:generalized_depolarization}, 
consider a scenario where the attacker's pre-masking posterior states evaluated at $O$ 
happen to reside within a single-qubit subspace as the canonical non-orthogonal states 
$\sigma_1 = \ket{+}\!\bra{+}$ and $\sigma_0 = \ket{0}\!\bra{0}$. The exact pre-masking 
trace distance is computed as $D(\sigma_1, \sigma_0) = \sqrt{1 - |\langle + | 0 \rangle|^2} 
= 1/\sqrt{2}$. 
Applying the scaling law, the closed-loop leakage under generalized depolarization 
is explicitly bounded and evaluated as
\[
L_{\rho_0}^{\pi_{\mathrm{ctrl}}}(O) = \frac{1-p}{\sqrt 2}.
\]
Crucially, this algebraic bound directly enables the analytic synthesis of 
an $\epsilon$-feasible enforcement policy (Definition~\ref{def:epsilon_feasible_enforcement}). 
Given a prescribed security threshold $\epsilon > 0$, the supervisor can mathematically 
guarantee compliance by selecting a minimal masking strength $p^\star$ that satisfies
\[
\frac{1-p^\star}{\sqrt 2} \le \epsilon \implies p^\star = \max\big(0,\, 1 - \sqrt{2}\epsilon\big).
\]
For instance, enforcing a rigorous leakage threshold of $\epsilon = 0.1$ analytically 
mandates a minimal depolarization intervention of $p^\star \approx 0.8586$. This 
demonstrates how the scaling law facilitates the direct translation of quantum 
distinguishability not just into theoretical bounds, but into precise, actionable 
control parameters for automated opacity enforcement.
\end{example}

\subsection{Counterexample-Guided Policy Improvement}

The exact verification engine naturally serves as the diagnostic subroutine within a 
policy-improvement loop. Starting from an empty enforcement policy, we iteratively 
verify the current closed-loop system and extract an observation pomset $O^\star$ 
causing an $\epsilon$-opacity violation, if one exists. Crucially, the nature of 
the opacity violation dictates the valid candidate policy updates, requiring a 
formal mechanism for local policy synthesis.

\subsubsection*{Local Policy Synthesis via Causal Hitting Sets}
Let $T_c(C) = \{t \in T_c \mid \exists e \in C \colon t_e = t\}$ denote the set of controllable 
transitions whose occurrences are present as events within a finite configuration $C$. 
The synthesis of local policy updates relies on the causal structure of the 
configurations generating the violating pomset $O^\star$:

\begin{itemize}[leftmargin=*]
    \item \textbf{Strict Flow Restriction ($\delta$-updates):} If $O^\star$ constitutes 
    a structural current-state opacity violation ($p(O^\star, 0) = 0$), it is exclusively 
    generated by the set of secret configurations $\mathcal{C}_1(O^\star)$. To eliminate 
    this violation, the supervisor must break every such causal chain. We formally define 
    $\mathcal{H}_{\min}(O^\star)$ as the family of all minimal hitting sets over the 
    collection $\{T_c(C) \mid C \in \mathcal{C}_1(O^\star)\}$. A valid local flow-control 
    update ($\Delta\delta$) is drawn strictly from $\mathcal{H}_{\min}(O^\star)$ to maximize 
    system permissiveness.
    
    \item \textbf{Quantum State Obfuscation ($\mu$-updates):} If the violation is purely 
    quantitative quantum leakage ($p(O^\star, 0) > 0$), the enforcement may leverage 
    invisible masking. A masking update injects a transition $t_m \in T_m$ uniformly across 
    $\mathcal{C}(O^\star)$, applying a local noise channel without altering the underlying 
    observation pomsets (i.e., preserving the structural visibility semantics).
    
    \item \textbf{Partial Flow Restriction (Quantitative $\delta$-updates):} As an alternative 
    resolution for quantitative quantum leakage, partial flow restriction can be applied using 
    any non-empty subset $\Delta\delta \subseteq \bigcup_{C \in \mathcal{C}_1(O^\star)} T_c(C)$ 
    to prune specific configuration subsets. This purposefully alters the convex combination of 
    the mixed states within the aggregated posterior, thereby reducing the trace distance without 
    injecting quantum noise.
\end{itemize}

When available, the reduced posterior ZX-certificates may be inspected as explanatory artifacts 
to heuristically guide the selection of these candidate updates. 
The integration of the exact verification subroutine and the aforementioned local synthesis rules 
is formally consolidated in Algorithm~\ref{alg:counterexample_guided_enforcement}. 
This iterative procedure systematically identifies violating pomsets, synthesizes candidate updates, 
filters them for admissibility against $\mathcal A_{\mathrm{spec}}$, and applies the optimal 
valid update until a certified closed-loop policy is achieved.

\begin{algorithm}[t]
\caption{Counterexample-Guided Policy Improvement}
\label{alg:counterexample_guided_enforcement}
\begin{algorithmic}[1]
\Require Controlled architecture $\Qsys^{\mathrm{ctrl}}$, leakage threshold $\epsilon \in [0,1)$, and iteration bound $K$
\Ensure An admissible $\epsilon$-feasible policy, if one is found within $K$ iterations
\State Initialize $\pi_{\mathrm{ctrl}}^{(0)} \gets (\emptyset, \emptyset)$ \Comment{Empty policy with no disabled/masking transitions}
\For{$k=0,\dots,K-1$}
    \State Run the exact verification engine on $\Qsys \parallel \pi_{\mathrm{ctrl}}^{(k)}$
    \If{$\sup_O L_{\rho_0}^{\pi_{\mathrm{ctrl}}^{(k)}}(O) \le \epsilon$}
        \State \Return $\pi_{\mathrm{ctrl}}^{(k)}$
    \EndIf
    
    \State Select an observation pomset $O^\star \in \Obs(\Qsys \parallel \pi_{\mathrm{ctrl}}^{(k)})$ such that $L_{\rho_0}^{\pi_{\mathrm{ctrl}}^{(k)}}(O^\star) > \epsilon$
    
    \State \textbf{Phase 1: Rigorous Local Policy Synthesis}
    \If{$p(O^\star, 0) = 0$} \Comment{Resolve structural current-state opacity violation}
        \State $\Pi_{\mathrm{cand}} \gets \big\{ (\Delta\delta, \emptyset) \;\big|\; \Delta\delta \in \mathcal{H}_{\min}(O^\star) \big\}$
    \Else \Comment{Resolve quantitative quantum leakage violation}
        \State $T_c(O^\star) \gets \bigcup_{C \in \mathcal{C}_1(O^\star)} T_c(C)$
        \State $\Pi_{\mathrm{cand}} \gets \big\{ (\emptyset, \{t_m\}) \;\big|\; t_m \in T_m \big\} \cup \big\{ (\Delta\delta, \emptyset) \;\big|\; \Delta\delta \subseteq T_c(O^\star) \setminus \{\emptyset\} \big\}$
    \EndIf
    
    \State \textbf{Phase 2: Admissibility Filtering}
    \State $\Pi_{\mathrm{adm}} \gets \big\{ \Delta\pi \in \Pi_{\mathrm{cand}} \;\big|\; \Qsys \parallel (\pi_{\mathrm{ctrl}}^{(k)} \cup \Delta\pi) \models \mathcal A_{\mathrm{spec}} \big\}$
    \If{$\Pi_{\mathrm{adm}} = \emptyset$}
        \State \Return \textbf{Failure}
    \EndIf
    
    \State \textbf{Phase 3: Cost-Optimal Policy Update}
    \State $\Delta\pi^\star \gets \arg\min_{\Delta\pi \in \Pi_{\mathrm{adm}}} \mathcal C_{\mathrm{cost}}(\pi_{\mathrm{ctrl}}^{(k)} \cup \Delta\pi)$
    \State $\pi_{\mathrm{ctrl}}^{(k+1)} \gets \pi_{\mathrm{ctrl}}^{(k)} \cup \Delta\pi^\star$
\EndFor
\State \Return \textbf{Failure}
\end{algorithmic}
\end{algorithm}

\begin{theorem}[Partial Correctness]
\label{thm:certified_enforcement}
If Algorithm~\ref{alg:counterexample_guided_enforcement} returns a policy
$\pi_{\mathrm{ctrl}}$, then $\pi_{\mathrm{ctrl}}$ is guaranteed to be admissible and
$\epsilon$-feasible.
\end{theorem}

\begin{proof}
By construction, any returned policy must pass the strict admissibility filter (Line 16), 
ensuring $\Qsys \parallel \pi_{\mathrm{ctrl}} \models \mathcal A_{\mathrm{spec}}$. 
Furthermore, the algorithm terminates with a successfully synthesized policy (Line 5) only when 
the exact verification engine certifies that both structural current-state opacity is satisfied 
and the supremum of the closed-loop quantum leakage across all reachable pomsets does not 
exceed the threshold $\epsilon$. Therefore, the returned policy is exactly admissible 
and $\epsilon$-feasible.
\end{proof}

\begin{remark}[Scope of the Enforcement Layer]
The enforcement layer is introduced here as an auxiliary framework to the primary 
verification theory. Its role is to provide a mathematically sound interface for 
counterexample-guided leakage reduction, explicitly formalizing policy synthesis via 
causal hitting sets and delineating strict flow restriction, partial flow restriction, 
and quantum state obfuscation. It is not intended to resolve the broader questions of globally 
optimal synthesis, completeness, or the formulation of a fully automated supervisory-control theory. 
These advanced synthesis challenges remain open for future investigation.
\end{remark}


\begin{figure*}[t]
\centering
\resizebox{0.85\linewidth}{!}{
\begin{tikzpicture}[>=stealth, thick]

    \tikzstyle{place} = [circle, draw=black, minimum size=9mm, inner sep=0pt, font=\large]
    \tikzstyle{place_b} = [circle, draw=blue, fill=blue!5, minimum size=9mm, inner sep=0pt, font=\large]
    \tikzstyle{place_r} = [circle, draw=red, fill=red!10, minimum size=9mm, inner sep=0pt, font=\large]
    \tikzstyle{place_f} = [circle, draw=black, fill=gray!10, minimum size=9mm, inner sep=0pt, font=\large]
    
    \tikzstyle{qubit} = [circle, draw=blue!80, fill=blue!5, dashed, minimum size=9mm, inner sep=0pt, font=\large]
    \tikzstyle{trans_obs} = [rectangle, draw=black, fill=gray!20, minimum width=1.4cm, minimum height=0.6cm, align=center]
    \tikzstyle{trans_unobs} = [rectangle, draw=black, fill=white, minimum width=1.4cm, minimum height=0.6cm, align=center]
    \tikzstyle{trans_sys} = [rectangle, draw=black, fill=yellow!20, minimum width=1.4cm, minimum height=0.6cm, align=center]
    \tikzstyle{q_access} = [->, dashed, color=blue!70, thick]
    \tikzstyle{token} = [fill=black, circle, inner sep=0pt, minimum size=4pt]

    
    \node[place, label=above:$p_{\mathrm{bg}}$] (pbg) at (1.5, 4.5) {};
    \node[token] at (pbg) {};
    \node[trans_obs, label=above:$\mathsf{cal}$] (tcal) at (3.5, 4.5) {$t_{\mathrm{cal}}$};
    \draw[->] (pbg) to[bend left=20] (tcal);
    \draw[->] (tcal) to[bend left=20] (pbg);

    \node[place, label=above:$p_0$] (p0) at (0, 0) {};
    \node[token] at (p0) {};
    \node[trans_obs, label=below:$\mathsf{req}$] (treq) at (1.5, 0) {$t_{\mathrm{req}}$};
    \node[place, label=above:$p_1$] (p1) at (3, 0) {};
    
    \draw[->] (p0) -- (treq);
    \draw[->] (treq) -- (p1);

    \node[trans_unobs, label=above:$\tau$] (tswap) at (4.5, 1.5) {$t_{\mathrm{swap}}^{\mathrm{non-sec}}$};
    \node[trans_unobs, label=below:$\tau$] (tpur) at (4.5, -1.5) {$t_{\mathrm{pur}}^{\mathrm{sec}}$};
    \draw[->] (p1) -- (tswap.west);
    \draw[->] (p1) -- (tpur.west);

    \node[place_b, label=above:$p_2^{\mathrm{non-sec}}$] (p2n) at (6.5, 1.5) {};
    \node[trans_obs, label=above:$\mathsf{swap\_ok}$] (tokn) at (8.5, 1.5) {$t_{\mathrm{ok}}^{\mathrm{non-sec}}$};
    \node[place_b, label=above:$p_3^{\mathrm{non-sec}}$] (p3n) at (10.5, 1.5) {};
    \node[trans_obs, label=above:$\mathsf{done}$] (tdonen) at (12.5, 1.5) {$t_{\mathrm{done}}^{\mathrm{non-sec}}$};
    
    \draw[->] (tswap) -- (p2n);
    \draw[->] (p2n) -- (tokn);
    \draw[->] (tokn) -- (p3n);
    \draw[->] (p3n) -- (tdonen);

    \node[place_r, label=below:$p_2^{\mathrm{sec}}$] (p2s) at (6.5, -1.5) {};
    \node[trans_obs, label=below:$\mathsf{swap\_ok}$] (toks) at (8.5, -1.5) {$t_{\mathrm{ok}}^{\mathrm{sec}}$};
    \node[place_r, label=below:$p_3^{\mathrm{sec}}$] (p3s) at (10.5, -1.5) {};
    \node[trans_obs, label=below:$\mathsf{done}$] (tdones) at (12.5, -1.5) {$t_{\mathrm{done}}^{\mathrm{sec}}$};
    \node[trans_obs, fill=red!20, label=right:$\mathsf{fail}$] (treject) at (6.5, -3.5) {$t_{\mathrm{reject}}$};
    
    \draw[->] (tpur) -- (p2s);
    \draw[->] (p2s) -- (toks);
    \draw[->] (toks) -- (p3s);
    \draw[->] (p3s) -- (tdones);
    \draw[->] (p2s) -- (treject);

    \node[place_f, label=below right:$p_{\mathrm{finish}}$] (pfinish) at (14, 0) {};
    \node[trans_sys, label=above:$\tau$] (treset) at (14, 3.5) {$t_{\mathrm{reset}}$};
    
    \draw[->, rounded corners=5pt] (tdonen.east) -- (13.4, 1.5) -- (13.4, 0.25) -- (pfinish.150);
    \draw[->, rounded corners=5pt] (tdones.east) -- (13.4, -1.5) -- (13.4, -0.25) -- (pfinish.210);
    \draw[->, rounded corners=5pt] (treject.east) -- (14, -3.5) -- (pfinish.south);
    
    \draw[->] (pfinish.north) -- (treset.south);
    
    \draw[->, rounded corners=8pt] (treset.west) -- (-0.8, 3.5) -- (-0.8, 0) -- (p0.west);

    \draw[dotted, rounded corners, draw=black!50] (-1.2, 5.5) rectangle (16.8, -4.5);
    \node[anchor=south west, color=black!70] at (-1, 5.6) {\textbf{Industrial Safe Control Net ($P_C, T, F_C$)}};

    \node[qubit] (q1) at (2.5, -8) {$q_1$};
    \node[qubit] (q2) at (5.5, -8) {$q_2$};
    \node[qubit] (q3) at (8.5, -8) {$q_3$};
    \node[qubit] (q4) at (11.5, -8) {$q_4$};
    \node[qubit, fill=red!10, draw=red, thick] (qM) at (7, -10) {$q_M$};
    \node[anchor=west, color=red] at (8, -10) {\textbf{Diagnostic Port (Attacker Interface $A$)}};

    \draw[thick, loosely dotted, color=black!60] (q1) -- (q2) node[midway, above, font=\small] {$\Phi^+$};
    \draw[thick, loosely dotted, color=black!60] (q3) -- (q4) node[midway, above, font=\small] {$\Phi^+$};

    \draw[dotted, rounded corners, draw=blue!50] (1.5, -6.5) rectangle (16.8, -12.5);
    \node[anchor=north west, color=blue!80] at (1.7, -12.6) {\textbf{Persistent Quantum Registers ($P_Q$)}};

    \draw[q_access] (tswap.south) to[out=-90, in=90] node[pos=0.6, left, font=\scriptsize, text=blue] {$\mathrm{BSM}$} (q2.north);
    \draw[q_access] (tswap.south) to[out=-45, in=90] (q3.north);
    \draw[q_access] (tokn.south) to[out=-90, in=90] node[pos=0.6, right, font=\scriptsize, text=blue] {$\mathrm{FF}$} (q4.north);

    \draw[q_access] (tpur.south) to[out=-90, in=110] (q2.north);
    \draw[q_access] (tpur.south) to[out=-45, in=110] (q3.north);
    \draw[q_access, color=red, thick] (tpur.south) to[out=-75, in=90] node[pos=0.8, left, font=\scriptsize, text=red] {$\mathrm{CNOT}$} (qM.north);
    \draw[q_access] (toks.south) to[out=-90, in=110] node[pos=0.6, left, font=\scriptsize, text=blue] {$\mathrm{FF}$} (q4.north);

    \draw[dotted, color=orange, thick] (treset.east) -- (16.2, 3.5) -- (16.2, -11.8) -- (2.5, -11.8);
    \node[font=\scriptsize, text=orange, align=left, anchor=south west] at (14.8, 3.6) {\textbf{Hardware}\\\textbf{Initialization ($\mathcal{E}_{\mathrm{reset}}$)}};
    
    \draw[->, dotted, color=orange, thick] (11.5, -11.8) -- (q4.south);
    \draw[->, dotted, color=orange, thick] (8.5, -11.8) -- (q3.south);
    \draw[->, dotted, color=orange, thick] (7, -11.8) -- (qM.south);
    \draw[->, dotted, color=orange, thick] (5.5, -11.8) -- (q2.south);
    \draw[->, dotted, color=orange, thick] (2.5, -11.8) -- (q1.south);
    
\end{tikzpicture}
} 
\caption{Representative SPO-QPN architecture for the entanglement-swapping repeater-service case study. 
The control layer contains a dual-lane foreground service, an observable background calibration 
self-loop, a unified completion place $p_{\mathrm{finish}}$, and an unobservable reset transition 
$t_{\mathrm{reset}}$. The reset implements the global replacement channel 
$\mathcal{E}_{\mathrm{reset}}$ that restores $\rho_0$ and returns the control token to $p_0$. 
The posterior analysis in Section~\ref{sec:case_repeater} is performed at the completed-round cut 
immediately before $t_{\mathrm{reset}}$ fires, thereby exposing the residual footprint on 
the diagnostic memory qubit $q_M$.}
\label{fig:spos_qpn_architecture}
\end{figure*}
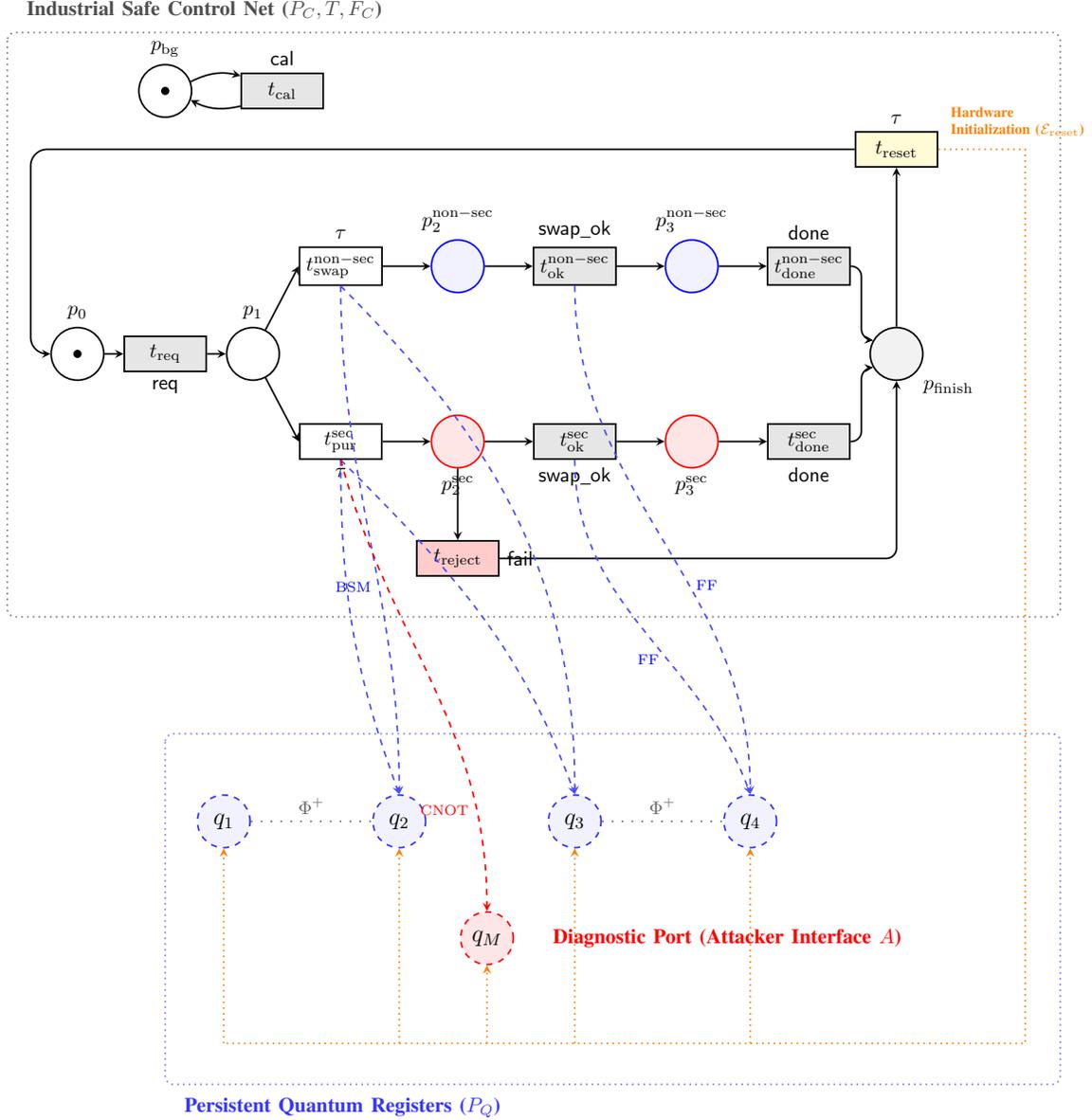

\section{Case Study: Entanglement-Swapping Service with Hidden Memory-Assisted Purification}
\label{sec:case_repeater}

Quantum repeater services provide a particularly stringent setting for opacity analysis: 
the service interface exposed to users is classical, while the physical resources that 
may retain hidden information are quantum and long-lived. As a result, 
a controller decision can remain unobservable at the level of completed service logs and 
still leave a distinguishable footprint on persistent memory. The repeater-service instance 
in Fig.~\ref{fig:spos_qpn_architecture} is designed to isolate exactly this mechanism within 
a cyclic architecture that also includes an observable abort lanes, 
a concurrent calibration loop, and a hardware-reset transition. 
This combination makes the instance suitable for examining, within one coherent SPO-QPN model, 
how exact posterior leakage analysis, masking-based enforcement, 
state-space reduction via true-concurrency semantics, 
and branch-local diagrammatic certification integrate together. The quantitative analysis below is 
performed on completed successful service rounds at the cut immediately before $t_{\mathrm{reset}}$, 
and the quantum operations are restricted to the stabilizer fragment so that all reported posteriors 
and certificates remain exact.

\subsection{Scenario and Secrecy Question}
\label{subsec:repeater_scenario}

\subsubsection{Operational Overview}
The case study considers a single foreground entanglement-swapping request while 
an observable background-calibration self-loop may fire independently of the foreground chain. 
The public service log exposes the observable alphabet
\[
\Sigma_{\mathrm{o}} \;=\; \{\mathsf{req},\;\mathsf{swap\_ok},\;\mathsf{done},\;\mathsf{cal},\;\mathsf{fail}\},
\]
where $\mathsf{cal}$ is emitted by repeated firings of the background-calibration self-loop 
and $\mathsf{fail}$ is emitted by the observable abort transition $t_{\mathrm{reject}}$ 
on the secret branch. 
Internally, the controller may route the request through a hidden memory-assisted 
purification branch, instantiated here by a hidden CNOT from $q_2$ to the recycled memory 
before the Bell-state measurement. While an aborting run trivially exposes the secret 
via the classical log, the core secrecy question addressed here is whether an attacker observing 
a seemingly successful public service log can still infer the execution of 
this hidden branch by reading the recycled memory qubit at the completed-round cut.

\subsubsection{Quantum registers and attacker interface}
The repeater system uses five persistent qubits
\[
q_1,\;q_2,\;q_3,\;q_4,\;q_M,
\]
where $(q_1,q_2)$ and $(q_3,q_4)$ form two Bell pairs and $q_M$ is the recycled memory qubit. 
The initial joint state is
\begin{equation}
\rho_0
=
\lvert\Phi^+\rangle\!\langle\Phi^+\rvert_{12}
\otimes
\lvert\Phi^+\rangle\!\langle\Phi^+\rvert_{34}
\otimes
\lvert 0\rangle\!\langle 0\rvert_M,
\label{eq:repeater_initial_state}
\end{equation}
with $\lvert\Phi^+\rangle=(\lvert00\rangle+\lvert11\rangle)/\sqrt{2}$. 
To formalize the threat model, the attacker's localized access is restricted strictly to 
this recycled memory, defining the attacker interface as $A=\{q_M\}$.

As formalized in Section~\ref{sec:opacity}, the evaluation of opacity in this framework relies on 
the configuration-based secret predicate $s(C)$. In this case study, an execution is designated 
as secret upon entering the place $p_2^{\mathrm{sec}}$. Therefore, the predicate $s(C)$ is 
instantiated to evaluate whether the configuration includes the corresponding hidden purification event:
\[
s(C)=1 \iff \exists e \in C : t_e = t_{\mathrm{pur}}^{\mathrm{sec}}.
\]
Consequently, the secret condition $b = s(C) = 1$ denotes configurations that 
traverse the hidden secret branch, and $b=0$ denotes those that use the non-secret lane.

\subsubsection{Formal safe control net}
The underlying classical net in Fig.~\ref{fig:spos_qpn_architecture} comprises places
\[
P_C \;=\; \{p_0,p_1,p_2^{\mathrm{non\text{-}sec}},p_3^{\mathrm{non\text{-}sec}},p_2^{\mathrm{sec}},p_3^{\mathrm{sec}},p_{\mathrm{finish}},p_{\mathrm{bg}}\},
\]
transitions
\[
\begin{aligned}
T=\{&t_{\mathrm{req}},t_{\mathrm{swap}}^{\mathrm{non\text{-}sec}},t_{\mathrm{ok}}^{\mathrm{non\text{-}sec}},t_{\mathrm{done}}^{\mathrm{non\text{-}sec}},t_{\mathrm{pur}}^{\mathrm{sec}},\\
&t_{\mathrm{ok}}^{\mathrm{sec}},t_{\mathrm{done}}^{\mathrm{sec}},t_{\mathrm{reject}},t_{\mathrm{cal}},t_{\mathrm{reset}}\},
\end{aligned}
\]
and the initial marking, formally represented as a subset of places for the safe net:
\begin{equation}
M_0 \;=\; \{p_0, p_{\mathrm{bg}}\}.
\label{eq:repeater_initial_marking}
\end{equation}
The single token on $p_{\mathrm{bg}}$ realizes the calibration self-loop 
$p_{\mathrm{bg}}\leftrightarrow t_{\mathrm{cal}}$, 
so that calibration may recur without constraining the foreground service. 
Physically, such a concurrent background loop models continuous asynchronous 
classical control traffic, such as routine hardware monitoring, heartbeat signals, 
or classical link calibration, running alongside the foreground quantum service.
The foreground flow is $p_0\xrightarrow{t_{\mathrm{req}}}p_1$, 
followed by the branching choice between $t_{\mathrm{swap}}^{\mathrm{non\text{-}sec}}$ 
and $t_{\mathrm{pur}}^{\mathrm{sec}}$ (manifesting as a structural conflict at $p_1$, which induces an immediate conflict between their corresponding events in the unfolding). 
The non-secret successful lane is
\[
p_2^{\mathrm{non\text{-}sec}}\xrightarrow{t_{\mathrm{ok}}^{\mathrm{non\text{-}sec}}}p_3^{\mathrm{non\text{-}sec}}\xrightarrow{t_{\mathrm{done}}^{\mathrm{non\text{-}sec}}}p_{\mathrm{finish}},
\]
whereas the secret successful lane is
\[
p_2^{\mathrm{sec}}\xrightarrow{t_{\mathrm{ok}}^{\mathrm{sec}}}p_3^{\mathrm{sec}}\xrightarrow{t_{\mathrm{done}}^{\mathrm{sec}}}p_{\mathrm{finish}}.
\]
In addition, the arc $p_2^{\mathrm{sec}}\xrightarrow{t_{\mathrm{reject}}}p_{\mathrm{finish}}$ 
represents an observable abort on the secret lane. Observation labels are as given in 
$\Sigma_{\mathrm{o}}$; the branch-selection transitions $t_{\mathrm{swap}}^{\mathrm{non\text{-}sec}}$, 
$t_{\mathrm{pur}}^{\mathrm{sec}}$, and the reset transition $t_{\mathrm{reset}}$ are unobservable.

\subsubsection{Hardware-initialization transition and evaluation cut}
Physical repeater hardware consumes entanglement upon each service round. To model this behavior, 
Fig.~\ref{fig:spos_qpn_architecture} includes a reset transition $t_{\mathrm{reset}}$, 
unobservable ($\tau$), which fires out of $p_{\mathrm{finish}}$ and returns the token to $p_0$. 
Its quantum access is global,
\[
\Acc(t_{\mathrm{reset}}) \;=\; \{q_1,q_2,q_3,q_4,q_M\},
\]
and its quantum effect is the replacement channel
\begin{equation}
\mathcal{E}_{\mathrm{reset}}(\rho) \;=\; \sum_{k}|\psi_0\rangle\!\langle k|\rho |k\rangle\!\langle\psi_0|
\;=\; \rho_0\,\mathrm{Tr}(\rho),
\label{eq:reset_channel}
\end{equation}
where $|\psi_0\rangle=|\Phi^+\rangle_{12}\otimes|\Phi^+\rangle_{34}\otimes|0\rangle_M$ 
and $\{|k\rangle\}$ is any orthonormal basis of the joint Hilbert space. 
The quantitative posterior analysis in this section is performed at the cut immediately 
before $t_{\mathrm{reset}}$ fires, when the control token resides at $p_{\mathrm{finish}}$. 
This is the earliest cut at which a service round has completed and 
the diagnostic readout of $q_M$ is operationally well-defined. 
Although both successful lanes converge at this identical completion place, 
their respective executions remain rigorously distinguished by 
the configuration-based predicate $s(C)$.
The inclusion of $t_{\mathrm{reset}}$ yields a cyclic service model; 
however, because the analysis is strictly bounded per completed round, 
the required finite unfolding prefix and the branch-local stabilizer calculations 
remain completely unaffected by this cyclic extension.

\subsubsection{Completed-round observation pomset}
Let $O_{\mathrm{fg}}$ denote the foreground observation pomset induced by
\[
\mathsf{req}\;\prec\;\mathsf{swap\_ok}\;\prec\;\mathsf{done}.
\]
This pomset corresponds to a completed successful service round. Aborting runs remain part of 
the full SPO-QPN through the observable event $\mathsf{fail}$, but they belong to observation pomsets 
distinct from $O_{\mathrm{fg}}$ and therefore do not enter the posterior aggregates conditioned on 
$O_{\mathrm{fg}}$.

\subsection{Branch Semantics and Observation Aggregation}
\label{subsec:repeater_branches}

\subsubsection{Non-secret lane ($b=0$)}
The unobservable branch-selection transition $t_{\mathrm{swap}}^{\mathrm{non\text{-}sec}}$ 
has the quantum access $\Acc(t_{\mathrm{swap}}^{\mathrm{non\text{-}sec}})=\{q_2,q_3\}$ and performs 
a Bell-state measurement (BSM). 
Physically, this operation applies the unitary transform
\begin{equation}
U_{\mathrm{BSM}}^{(q_2,q_3)} \;=\; (H_{q_2}\otimes I_{q_3})\,\mathrm{CNOT}_{q_2\rightarrow q_3}
\label{eq:bsm_transform}
\end{equation}
(where $H$ and $\mathrm{CNOT}$ denote the Hadamard and controlled-NOT gates, respectively), 
followed by a projective measurement of $(q_2,q_3)$ in the computational basis. 
The classical measurement outcomes 
$(m_2,m_3)\in\{0,1\}^2$ directly constitute the formal branch 
$r \in R_{t_{\mathrm{swap}}^{\mathrm{non\text{-}sec}}}$.
The downstream transition $t_{\mathrm{ok}}^{\mathrm{non\text{-}sec}}$ 
has $\Acc(t_{\mathrm{ok}}^{\mathrm{non\text{-}sec}})=\{q_4\}$ and performs 
the classically conditioned feed-forward (FF) Pauli correction
\begin{equation}
X_{q_4}^{m_3} Z_{q_4}^{m_2}.
\label{eq:feedforward}
\end{equation}

\subsubsection{Secret lane ($b=1$)}
As shown in Fig.~\ref{fig:spos_qpn_architecture}, the unobservable transition 
$t_{\mathrm{pur}}^{\mathrm{sec}}$ has the enlarged quantum access 
$\Acc(t_{\mathrm{pur}}^{\mathrm{sec}})=\{q_2,q_3,q_M\}$; 
the CNOT arc onto $q_M$ represents the hidden coupling, while the unlabeled arcs onto $q_2,q_3$ 
indicate the requisite access for the BSM, which is merged into the same transition on this lane. 
Formally, $t_{\mathrm{pur}}^{\mathrm{sec}}$ applies the compound unitary
\begin{equation}
U_{\mathrm{pur}+\mathrm{BSM}}
\;=\;
U_{\mathrm{BSM}}^{(q_2,q_3)} \, \mathrm{CNOT}_{q_2\to q_M},
\label{eq:pur_bsm_compound}
\end{equation}
followed by the same projective measurement of $(q_2,q_3)$. 
The resulting outcomes $(m_2,m_3)$ similarly define 
the formal branch $r \in R_{t_{\mathrm{pur}}^{\mathrm{sec}}}$.
The feed-forward correction $X_{q_4}^{m_3}Z_{q_4}^{m_2}$ is then applied by 
$t_{\mathrm{ok}}^{\mathrm{sec}}$, with $\Acc(t_{\mathrm{ok}}^{\mathrm{sec}})=\{q_4\}$. 
For executions reaching the completed-round cut, the hidden coupling $\mathrm{CNOT}_{q_2\to q_M}$ 
constitutes the sole quantum-operational difference between the two lanes; the BSM and feed-forward 
protocols are otherwise identical.

\subsubsection{Observation aggregation}
The branch indices $r=(m_2,m_3)$ remain internal to the classical controller; thus, each of 
the four branches emits the identical observable label $\mathsf{swap\_ok}$. 
Under the event-structure semantics, the four measurement branches on each lane induce 
the identical public observation pomset $O_{\mathrm{fg}}$. When computing the attacker's 
posterior, their unnormalized conditional density matrices are therefore summed prior to 
normalization. For completed successful rounds, the only quantum difference between $b=0$ and $b=1$ 
is the hidden coupling left on $q_M$ after this aggregation; the observable abort transition 
$t_{\mathrm{reject}}$ belongs to observation classes distinct from $O_{\mathrm{fg}}$.

\subsection{Exact Posterior Leakage in the Base Instance}
\label{subsec:repeater_exact_results}

\subsubsection{Non-secret posterior}
For $b=0$, the recycled memory is never coupled to the service qubits, so the marginal state on 
$q_M$ is unchanged from \eqref{eq:repeater_initial_state}:
\begin{equation}
\bar{\sigma}_{O_{\mathrm{fg}},0}(\rho_0) \;=\; \lvert 0\rangle\!\langle 0\rvert_M.
\label{eq:posterior_nonsecret}
\end{equation}

\subsubsection{Secret posterior}
For $b=1$, applying $\mathrm{CNOT}_{q_2\to q_M}$ (the hidden coupling component of 
\eqref{eq:pur_bsm_compound}) to the Bell-pair half on $q_2$ and to $|0\rangle_M$ yields 
the Greenberger-Horne-Zeilinger (GHZ) state $(|000\rangle+|111\rangle)/\sqrt{2}$ on 
the subspace $(q_1,q_2,q_M)$. Although $q_M$ is now entangled with $q_1$ and $q_2$, 
the subsequent BSM on $(q_2,q_3)$ and feed-forward on $q_4$ act strictly on subsystems disjoint 
from $q_M$. Since the four branch outcomes of the BSM are aggregated to form the public 
observation (constituting a non-selective quantum operation), these downstream local operations 
leave the reduced density matrix of $q_M$ strictly invariant. Consequently, tracing out the full 
system $(q_1,q_2,q_3,q_4)$ at the evaluation cut is mathematically equivalent to tracing out 
$(q_1,q_2)$ immediately after the hidden $\mathrm{CNOT}_{q_2\to q_M}$ coupling:
\begin{equation}
\bar{\sigma}_{O_{\mathrm{fg}},1}(\rho_0) \;=\; \mathrm{Tr}_{q_1,q_2}\left[ \frac{(|000\rangle+|111\rangle)(\langle 000|+\langle 111|)}{2} \right] \;=\; \frac{I_2}{2}.
\label{eq:posterior_secret}
\end{equation}

\subsubsection{Exact trace-distance leakage}
By Definition~\ref{def:posterior_state_leakage} (trace-distance between conditional posteriors),
\begin{equation}
\begin{aligned}
L_{\rho_0}(O_{\mathrm{fg}}) &= \tfrac{1}{2}\bigl\lVert \bar{\sigma}_{O_{\mathrm{fg}},1}(\rho_0) - \bar{\sigma}_{O_{\mathrm{fg}},0}(\rho_0) \bigr\rVert_1 \\
&= \tfrac{1}{2}\bigl\lVert \tfrac{I_2}{2} - \lvert 0\rangle\!\langle 0\rvert_M \bigr\rVert_1 \;=\; 0.5.
\end{aligned}
\label{eq:base_leakage}
\end{equation}
The base instance therefore violates posterior-state opacity with an exact leakage value of 
$0.5$, as illustrated by Fig.~\ref{fig:repeater_posteriors}.

\begin{figure}[t]
    \centering
    \includegraphics[width=0.88\linewidth]{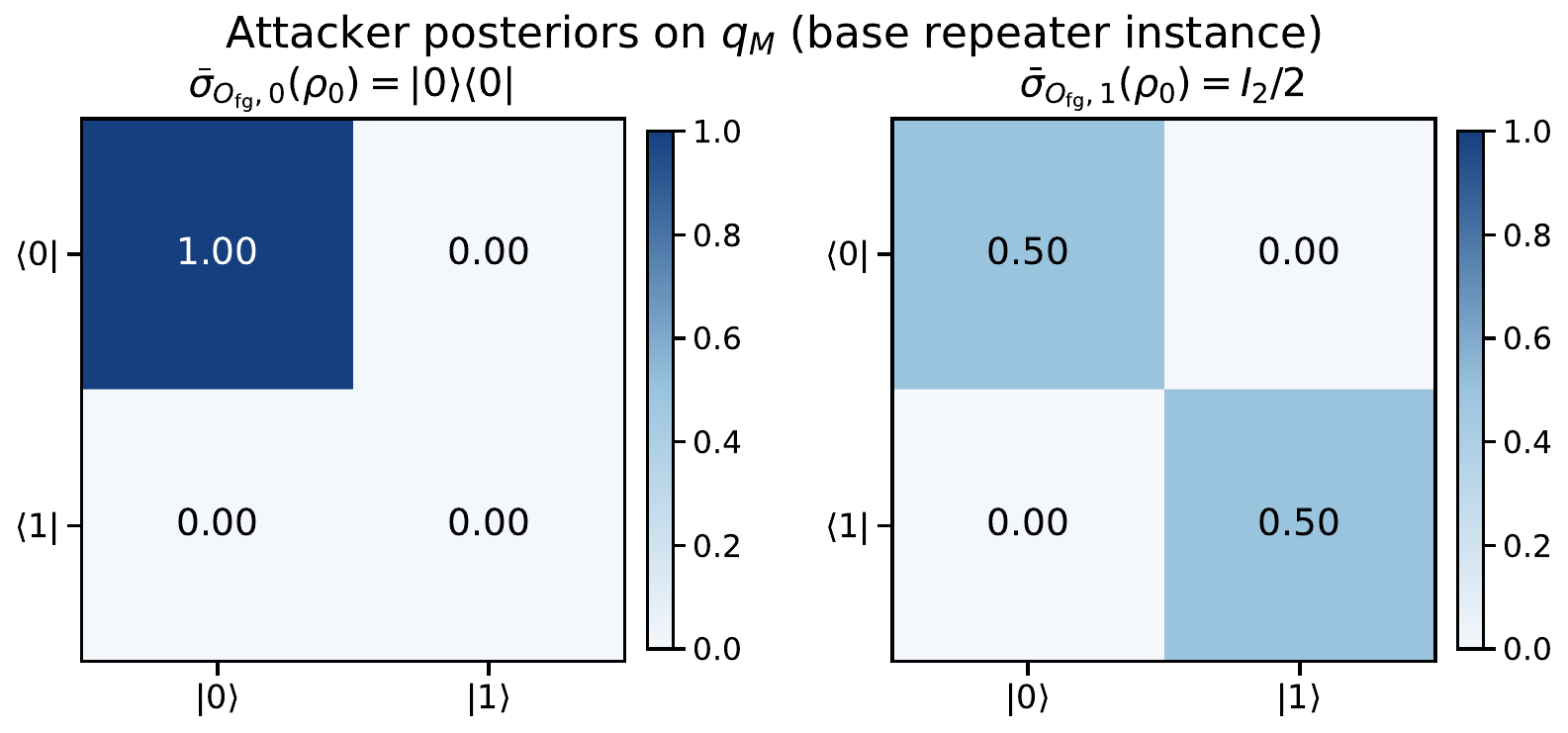}
    \caption{Exact secret-conditioned attacker posteriors on the recycled memory qubit ($q_M$) 
    for the base repeater-service instance. The non-secret execution yields the pure state 
    $\bar{\sigma}_{O_{\mathrm{fg}},0}(\rho_0)=\lvert 0\rangle\!\langle 0\rvert_M$, 
    while the secret execution yields the maximally mixed state 
    $\bar{\sigma}_{O_{\mathrm{fg}},1}(\rho_0)=I_2/2$. Since both posteriors are conditioned on 
    the identical public observation pomset $O_{\mathrm{fg}}$, the trace distance between them 
    strictly quantifies the posterior-state leakage.}
    \label{fig:repeater_posteriors}
\end{figure}

\subsection{Closed-Loop Enforcement via Invisible Depolarizing Masking ($\mu$-update)}
\label{subsec:repeater_repair}

To rigorously mitigate the identified violation, we synthesize a local policy update utilizing 
the quantum state obfuscation mechanism ($\mu$-update) formalized in Section~\ref{sec:enforcement}. 
Specifically, the supervisor injects an invisible masking transition acting strictly 
on the attacker interface $q_M$ after service completion, implementing 
the generalized single-qubit depolarizing channel 
(as defined in Proposition~\ref{prop:generalized_depolarization}):
\begin{equation}
\mathcal{E}_p(\omega)=(1-p)\omega + p\frac{I_2}{2}, \qquad p\in[0,1].
\label{eq:depolarizing_channel}
\end{equation}
Since $\mathcal{E}_p(I_2/2)=I_2/2$, the secret posterior remains mathematically invariant under 
this channel ($\bar{\sigma}^{\pi_{\mathrm{ctrl}}}_{O_{\mathrm{fg}},1}(\rho_0) = I_2/2$), 
while the modified non-secret posterior updates to:
\[
\bar{\sigma}^{\pi_{\mathrm{ctrl}}}_{O_{\mathrm{fg}},0}(\rho_0) = \mathcal{E}_p(\lvert 0\rangle\!\langle 0\rvert_M) = (1-p)\lvert 0\rangle\!\langle 0\rvert_M + p\frac{I_2}{2}.
\]
Consequently, by the contractivity of trace distance, the closed-loop leakage algebraically scales to:
\begin{equation}
L_{\rho_0}^{\pi_{\mathrm{ctrl}}}(O_{\mathrm{fg}})
=
\frac{1}{2}\left\lVert \frac{I_2}{2} - \bar{\sigma}^{\pi_{\mathrm{ctrl}}}_{O_{\mathrm{fg}},0}(\rho_0) \right\rVert_1
=
\frac{1-p}{2}.
\label{eq:repeater_repaired_leakage}
\end{equation}
Given a strict $\epsilon$-feasible enforcement threshold $L_{\rho_0}^{\pi_{\mathrm{ctrl}}}(O_{\mathrm{fg}}) \leq \epsilon$, the minimal masking strength analytically required by the supervisor is:
\begin{equation}
 p^\star = 1-2\epsilon.
 \label{eq:pstar_repeater}
\end{equation}
Assuming a stringent security requirement of $\epsilon=0.05$ for the base setting, 
the synthesized policy mandates $p^\star=0.9$; the resulting analytic enforcement trajectory 
is plotted in Fig.~\ref{fig:repeater_repair_curve}.

\begin{figure}[t]
    \centering
    \includegraphics[width=0.78\linewidth]{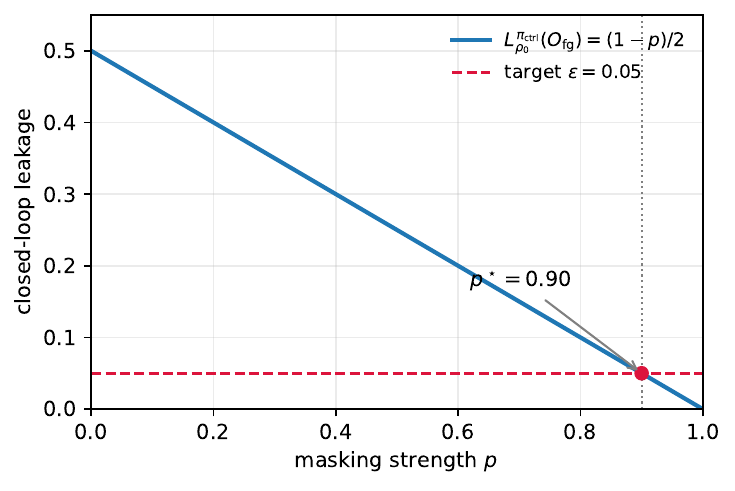}
    \caption{Analytic policy-enforcement curve for the base instance under 
    an invisible depolarizing $\mu$-update. The closed-loop leakage decreases linearly following 
    $L_{\rho_0}^{\pi_{\mathrm{ctrl}}}(O_{\mathrm{fg}})=(1-p)/2$. 
    The horizontal dashed line denotes the target security threshold $\epsilon=0.05$, 
    intercepting the minimal feasible intervention parameter at $p^\star=0.9$.}
    \label{fig:repeater_repair_curve}
\end{figure}

\subsection{Concurrency Quotient on a Bounded Service Family}
\label{subsec:repeater_scaling}

To expose the effect of the concurrency quotient in isolation, we introduce a bounded family of 
execution scenarios $\{E_m\}_{m\geq 0}$ indexed by the number $m$ of background-calibration firings 
during a single successful foreground service round. Formally, $E_m$ restricts attention to 
configurations in which $t_{\mathrm{cal}}$ fires exactly $m$ times before the first 
$t_{\mathrm{reset}}$ and the observable target is the successful completed-round pomset 
$O_{\mathrm{fg}}^{(m)}$. The full model still contains the abort label $\mathsf{fail}$, but those 
executions belong to observation classes distinct from $O_{\mathrm{fg}}^{(m)}$ and therefore do not 
enter this family.

The family is chosen so that the calibration loop remains structurally disjoint from the foreground 
service. This controlled setting cleanly isolates the semantic role of the event-structure quotient. 
The space of 
observable linear interleavings then grows as the set of valid shuffles of $m$ 
occurrences of the label $\mathsf{cal}$ with 
the three observable labels of the successful foreground chain:
\begin{equation}
N_{\mathrm{seq}}(m)\;=\;\binom{m+3}{3}.
\label{eq:word_count_repeater}
\end{equation}
We evaluate the runtime scaling across the parameterized family $E_m$ to empirically demonstrate 
the computational advantage of the event-structure quotient. 
The experiments compare a naive interleaving-based exact simulator against 
the true-concurrency quotient exploration. To strictly isolate the semantic acceleration of 
the quotient from the algebraic optimizations of the symbolic stabilizer backend utilized elsewhere, 
both approaches are evaluated over an identical dense-matrix backend.

Measured execution times are summarized in Table~\ref{tab:repeater_scaling}. 
As the number of background calibration events $m$ increases, 
the interleaving baseline suffers a severe combinatorial explosion, 
driven by the $\binom{m+3}{3}$ observable interleavings. 
In stark contrast, the quotient-aware approach collapses these redundant interleavings into 
a single observation pomset. While the evaluation depth still scales with $m$, 
the true-concurrency exploration successfully neutralizes the interleaving-induced explosion, 
achieving an $82.83\times$ speedup at $m=12$. This controlled setting explicitly confirms 
that the event-structure quotient translates directly into substantial computational savings. 
The prototype implementation and reproducibility scripts are available 
in the public repository~\cite{Ding2026SPOQPN}.

\begin{table}[t]
    \centering
    \caption{Representative runtimes for the bounded family $E_m$ on a common dense-matrix backend. 
    The table isolates the effect of the event-structure quotient on the disjoint calibration loop while keeping 
    the quantum backend fixed.}
    \label{tab:repeater_scaling}
    \resizebox{\columnwidth}{!}{%
    \begin{tabular}{ccccc}
        \toprule
        $m$ & $N_{\mathrm{seq}}(m)$ & Interleaving-Dense (ms) & Quotient-Dense (ms) & Speedup \\
        \midrule
        0  & 1 & 107.028 & 102.016 & 1.05 \\
        4  & 35 & 1120.850 & 160.754 & 6.97 \\
        8  & 165 & 8340.708 & 249.208 & 33.47 \\
        12  & 455 & 36196.499 & 436.991 & 82.83 \\
        \bottomrule
    \end{tabular}%
}
\end{table}

\subsection{Secret-Conditioned Posterior Certificates}
\label{subsec:repeater_zx}

While the bounded execution family evaluated the quotient mechanism, the base repeater instance 
also serves to illustrate the diagrammatic certification layer. To connect the algorithmic output 
with the framework of Section~\ref{sec:zx_certificates}, we synthesize the unnormalized exact 
posteriors of the base instance into restricted, secret-conditioned posterior certificates. 
The purpose is not to encode the full repeater network, but to visually isolate the 
structural distinction responsible for the leakage and to express it through compact, 
independently verifiable topological objects.

For $b=0$, the certificate reduces via standard rewrite rules to the normal form 
$\mathrm{NF}_{\mathrm{zero}}$, a pure-$\lvert 0\rangle$ diagrammatic object. Conversely, 
for $b=1$, the certificate evaluates via environmental discard rules into the mixed normal form 
$\mathrm{NF}_{\mathrm{mixed}}$. The graphical derivation for the secret execution is given in 
Fig.~\ref{fig:repeater_zx_proof}, and the final canonical normal forms are shown in 
Fig.~\ref{fig:canonical_certificate}. 

In contrast to the zero-leakage equivalence established 
in Corollary~\ref{cor:zx_zero_leakage_certificate}, 
their role here is explicitly diagnostic: they resolve to distinct canonical normal forms that 
diagrammatically isolate the structural mechanism of the algebraic leakage, 
despite originating from 
the identical foreground observation pomset.

\begin{figure*}[t]
\centering

\resizebox{\textwidth}{!}{ 
\begin{minipage}[b]{0.3\linewidth}
\centering
\begin{tikzpicture}
  \node[x spider, label=left:{$|0\rangle$}] (init) at (0, 0) {};
  \node[x spider] (target) at (1.5, 0) {};
  \coordinate (out) at (3, 0);
  \draw[thick] (init) -- (target) -- (out);
  \node[anchor=south] at (out) {$q_M$};

  \coordinate (env_in) at (0, 1.5);
  \node[z spider] (control) at (1.5, 1.5) {};
  \coordinate (discard_q2) at (3, 1.5);
  \draw[thick] (env_in) -- (control) -- (discard_q2);
  
  \draw[thick] (control) -- (target);

  \zxdiscard{env_in}
  \zxdiscard{discard_q2}
\end{tikzpicture}
\subcaption{1. Initial Circuit}
\end{minipage}
\hfill
$\xrightarrow{\text{Spider Fusion}}$
\hfill
\begin{minipage}[b]{0.3\linewidth}
\centering
\begin{tikzpicture}
  \node[x spider] (merged) at (1.5, 0) {};
  \coordinate (out) at (3, 0);
  \draw[thick] (merged) -- (out);
  \node[anchor=south] at (out) {$q_M$};

  \coordinate (env_in) at (0, 1.5);
  \node[z spider] (control) at (1.5, 1.5) {};
  \coordinate (discard_q2) at (3, 1.5);
  \draw[thick] (env_in) -- (control) -- (discard_q2);
  \draw[thick] (control) -- (merged);

  \zxdiscard{env_in}
  \zxdiscard{discard_q2}
\end{tikzpicture}
\subcaption{2. Merged Red Spiders}
\end{minipage}
\hfill
$\xrightarrow{\text{Discard Rule}}$
\hfill
\begin{minipage}[b]{0.3\linewidth}
\centering
\begin{tikzpicture}
  \coordinate (discard_final) at (1.2, 0);
  \coordinate (out) at (2.5, 0);
  \draw[thick] (discard_final) -- (out);
  \zxdiscard{discard_final}
  \node[anchor=south] at (out) {$q_M$};
\end{tikzpicture}
\subcaption{3. Maximally Mixed State}
\end{minipage}
}
\caption{Graphical algebraic derivation of the posterior state for the secret branch ($b=1$). 
Through sound rewrites in the mixed-state stabilizer ZX-calculus, the hidden CNOT coupling 
collapses to the maximally mixed state. Physically, this occurs because the control qubit ($q_2$), 
which is inherently maximally mixed due to its initial entanglement with the unobservable environment, 
is subsequently traced out via the Discard rules.}
\label{fig:repeater_zx_proof}
\end{figure*}
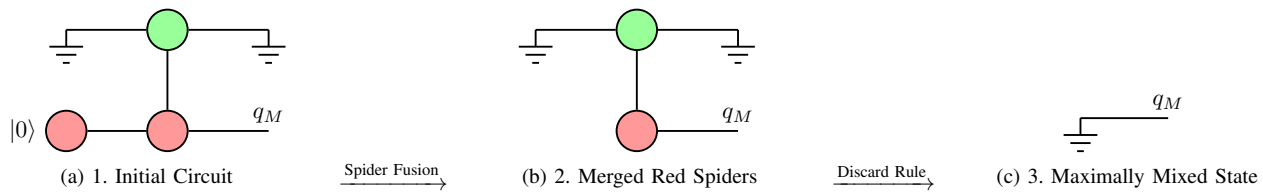

\begin{figure}[t]
\centering

\begin{minipage}[b]{0.45\linewidth}
\centering
\begin{tikzpicture}[scale=1.5]
  \node[x spider, label=left:{$|0\rangle$}] (pure_out) at (0, 0) {};
  \coordinate (out) at (1.5, 0);
  \draw[thick] (pure_out) -- (out);
  
  \node[anchor=south, font=\Large] at (out) {$q_M$};
  \node[anchor=north] at (0.75, -0.4) {\textbf{Canonical $\mathrm{NF}_{\mathrm{zero}}$ ($|0\rangle\langle0|$)}};
\end{tikzpicture}
\subcaption{Certificate A: Non-Secret Branch ($b=0$)}
\end{minipage}
\hfill
\begin{minipage}[b]{0.45\linewidth}
\centering
\begin{tikzpicture}[scale=1.5]
  \coordinate (discard_start) at (0, 0);
  \coordinate (out) at (1.5, 0);
  \draw[thick] (discard_start) -- (out);
  \zxdiscard{discard_start}
  
  \node[anchor=south, font=\Large] at (out) {$q_M$};
  \node[anchor=north] at (0.75, -0.4) {\textbf{Canonical $\mathrm{NF}_{\mathrm{mixed}}$ ($I/2$)}};
\end{tikzpicture}
\subcaption{Certificate B: Secret Branch ($b=1$)}
\end{minipage}

\caption{Consistent canonical posterior certificates (Normal Forms). The distinct topological structures of these two graphs provide a visually verifiable diagnostic of the posterior-state opacity violation at the attacker interface.}
\label{fig:canonical_certificate}
\end{figure}
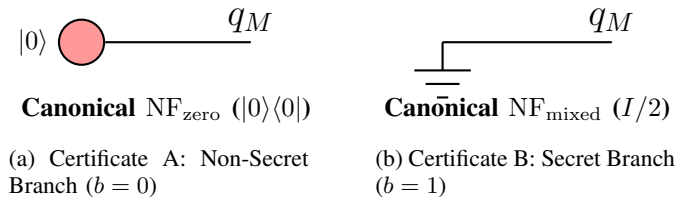

\subsection{Discussion}
\label{subsec:repeater_discussion}

This case study supports the main claims of the paper through a unified, causally coherent 
service model.
At the behavioral level, the completed successful round exposes the identical public observation 
pomset across both executions, while the corresponding attacker posteriors differ as derived 
in \eqref{eq:posterior_nonsecret} and \eqref{eq:posterior_secret}. The instance therefore 
demonstrates that observation equivalence within the observable event structure does not preclude 
quantum distinguishability at the attacker interface.

At the quantitative level, the leakage is exactly computable and admits a closed-form 
mitigation law. The base instance yields $L_{\rho_0}(O_{\mathrm{fg}})=0.5$, and the invisible 
depolarizing $\mu$-update reduces this value according to \eqref{eq:repeater_repaired_leakage}, 
with the minimal feasible intervention parameter analytically given by \eqref{eq:pstar_repeater}.

At the algorithmic level, the parameterized family $E_m$ confirms that the event-structure 
quotient fundamentally neutralizes the combinatorial explosion induced by observable interleavings. 
By collapsing these redundant interleavings into a single observation pomset, 
the evaluation strictly avoids the polynomial blow-up without altering the exact posterior aggregates.
This empirically validates the computational advantage of the true-concurrency semantics.

At the certification level, the algebraic posterior separation admits compact, secret-conditioned 
posterior certificates. These topological objects provide a visually verifiable diagnostic of 
the distinction detected by the symbolic engine, directly connecting the exact leakage computation 
with the normal-form structures developed in Section~\ref{sec:zx_certificates}.

Taken together, the repeater-service instance traces a coherent path from partially observed control 
logic to exact posterior leakage and then to independently verifiable diagrammatic certificates. 
Concurrently, the parameterized family $E_m$ demonstrates how true-concurrency semantics removes 
spurious interleaving artifacts, guaranteeing algorithmic tractability for the quantitative analysis.

\section{Conclusion}
\label{sec:conclusion}

This paper has established a concurrency-aware verification and enforcement framework for 
quantum opacity in partially observed systems, formally grounded in SPO-QPNs. 
The overarching contribution lies in demonstrating that true concurrency, 
partial observation, quantum side information, exact symbolic reasoning, diagrammatic certification, 
and closed-loop policy enforcement can be seamlessly integrated into a unified, 
mathematically rigorous architecture.

The architectural design of SPO-QPNs provides the necessary structural foundation. 
Under an event-structure semantics, the use of observation pomsets fundamentally neutralizes 
the combinatorial state-space explosion induced by observable interleavings, 
thereby eliminating the spurious security distinctions between structurally equivalent executions. 
Furthermore, the formulation of quantitative quantum opacity over normalized posterior states yields 
a rigorous, operationally meaningful leakage metric. To make this theory computationally tractable 
and analytically exact, we developed a true-concurrency symbolic aggregation engine for 
the stabilizer fragment. This algorithmic core is complemented by a ZX-calculus layer that 
resolves unnormalized algebraic posteriors into compact, 
independently verifiable diagrammatic certificates. Moving from verification to synthesis, 
the counterexample-guided enforcement loop connects the exact detection of opacity violations to 
certified closed-loop mitigation via classical flow restriction ($\delta$-updates) and 
invisible quantum state obfuscation ($\mu$-updates).

While the resulting framework is intentionally scoped to the stabilizer fragment, 
it remains structurally expressive: it successfully captures and mitigates genuinely 
concurrent quantum leakage phenomena while preserving exact computability and algorithmic efficiency. 
The empirical evaluation confirms that the true-concurrency semantics translates directly into 
substantial computational savings over interleaving-based baselines. 
To support independent reproducibility, the prototype implementation and 
associated evaluation scripts are available in the public artifact repository~\cite{Ding2026SPOQPN}. 
Extending this framework to encompass dynamic register allocation, channel-level opacity metrics, 
and decentralized quantum supervisory control forms a natural trajectory for future work. 
Ultimately, this precise balance between semantic expressive power and 
computational feasibility provides a robust foundation for the broader theory of 
information-flow security, formal verification, and automated synthesis 
in next-generation concurrent quantum systems.

\bibliographystyle{IEEEtran}

\bibliography{tmp}

\end{document}